\documentclass[english]{IEEEtran}

\pdfoutput=1

\usepackage[T1]{fontenc}
\usepackage[latin9]{inputenc}
\usepackage{verbatim}
\usepackage{float}
\usepackage{amsthm}
\usepackage{amsmath}
\usepackage{amssymb}
\usepackage{graphicx}
\usepackage{color}
\usepackage{url}

\usepackage{xr}

\makeatletter

\floatstyle{ruled}
\newfloat{algorithm}{tbp}{loa}
\providecommand{\algorithmname}{Algorithm}
\floatname{algorithm}{\protect\algorithmname}

\numberwithin{equation}{section}
\numberwithin{figure}{section}
\theoremstyle{plain}
\newtheorem{thm}{\protect\theoremname}[section]
\theoremstyle{definition}

\theoremstyle{remark}

\theoremstyle{plain}
\newtheorem{lem}[thm]{\protect\lemmaname}

\newtheorem*{lem*}{Lemma}

\theoremstyle{remark}

\theoremstyle{plain}

\theoremstyle{plain}
\newtheorem{proposition}[thm]{\protect\propositionname}

\@ifundefined{showcaptionsetup}{}{%
 \PassOptionsToPackage{caption=false}{subfig}}
\usepackage{subfig}
\makeatother

\usepackage{babel}
\providecommand{\claimname}{Claim}
\providecommand{\definitionname}{Definition}
\providecommand{\lemmaname}{Lemma}
\providecommand{\remarkname}{Remark}
\providecommand{\theoremname}{Theorem}
\providecommand{\corollaryname}{Corollary}
\providecommand{\propositionname}{Proposition}

\newcommand{\reals}{\mathbb{R}}
\newcommand{\RN}{\mathbb{R}^N}
\newcommand{\CN}{\mathbb{C}^N}
\newcommand{\RNN}{\mathbb{R}^{N\times N}}
\newcommand{\CNN}{\mathbb{C}^{N\times N}}
\newcommand{\inner}[2]{\left\langle {#1}, {#2} \right\rangle}
\newcommand{\adj}{^\mathrm{adj}}
\newcommand{\transpose}{^\top\! }
\newcommand{\sign}{\mathrm{phase}}
\newcommand{\trace}{\mathrm{Tr}}
\newcommand{\diag}{\mathrm{diag}}
\newcommand{\calM}{\mathcal{M}}
\newcommand{\T}{\mathrm{T}}
\newcommand{\Proj}{\operatorname{Proj}}
\newcommand{\ddiag}{\operatorname{ddiag}}
\newcommand{\D}{\operatorname{D}\!}
\newcommand{\grad}{\operatorname{grad}\!}
\newcommand{\Hess}{\operatorname{Hess}\!}
\newcommand{\Retr}{\operatorname{Retr}}

\newcommand{\DFT}[1]{\operatorname{DFT}\!\left( {#1} \right)}

\newcommand{\dist}{\operatorname{dist}}

\begin{document}


\title{Bispectrum Inversion\\with Application to Multireference Alignment}


\author{Tamir Bendory$^{\dagger}$\thanks{$\dagger$ These two authors contributed equally to the work.}, Nicolas Boumal$^{\dagger}$, Chao Ma, Zhizhen Zhao  and Amit Singer\thanks{TB, NB, CM and AS are with Princeton University in PACM and the Mathematics Department. ZZ is with the University of Illinois at Urbana-Champaign in the Department of Electrical and Computer Engineering. {The authors were partially supported by Award Number R01GM090200 from the NIGMS, FA9550-17-1-0291 from AFOSR, Simons Investigator Award and Simons Collaboration on Algorithms and Geometry from Simons Foundation, and the Moore Foundation Data-Driven Discovery Investigator Award. NB is partially supported by NSF grant DMS-1719558. ZZ is partially supported by National Center for Supercomputing Applications Faculty Fellowship and University of Illinois at Urbana-Champaign College of Engineering Strategic Research Initiative. }}}
\maketitle

\begin{abstract}

We consider the problem of estimating a signal from  noisy circularly-translated  versions of itself, called \emph{multireference alignment} (MRA).  
One natural approach to MRA could be to estimate the shifts of the observations first, and  infer the signal by aligning and averaging the data.
In contrast, we consider a method based on estimating the signal directly, using features of the signal that are invariant under translations. 
Specifically, we estimate the power spectrum and the bispectrum of the signal from the observations. Under mild assumptions, these invariant features contain enough information to infer the signal.
In particular, the bispectrum can be used to estimate the Fourier phases. To this end, we propose and analyze a few algorithms. {Our main  methods consist of  non-convex optimization  over the smooth manifold of phases. Empirically, in the absence of noise, these non-convex algorithms appear to converge to the target signal with random initialization. The algorithms are also robust to noise.}
 We then suggest three additional methods. These methods are based on frequency marching, semidefinite relaxation and integer programming. The first two methods provably recover the phases exactly in the absence of noise. In the high noise level regime, the invariant features approach for MRA results in stable estimation if the number of measurements scales like the cube of the noise variance, which is {the information-theoretic rate.} Additionally, it requires only one pass over the data which is important at low signal--to--noise ratio when the number of observations must be large.

\end{abstract}


\section{Introduction} \label{sec:intro}






We consider the problem of estimating a discrete signal from multiple noisy and translated  (i.e., circularly shifted) versions of itself, called \emph{multireference aligment} (MRA). This problem occurs in a variety of applications in biology \cite{diamond1992multiple,theobald2012optimal,park2011stochastic,park2014assembly}, radar \cite{zwart2003fast,gil2005using}, image registration   and super-resolution  \cite{dryden1998statistical,foroosh2002extension,robinson2009optimal}, and has been the subject of recent theoretical analysis \cite{abbe2017sample,bandeira2017optimal}. The MRA model reads 
\begin{align}
	{\xi}_j & = R_{r_j}{x}+\varepsilon_j,\quad j=1,\dots,M,
	\label{eq:MRA}
\end{align}
where $\varepsilon_j$ are i.i.d. normal random vectors with variance $\sigma^2$ and the underlying signal $x$ is in $\mathbb{R}^N$ or in $\mathbb{C}^N$. Operator $R_{r_j}$ rotates the signal ${x}$ circularly by $r_j$ locations, namely, $(R_{r_j}{x})[n]={x}[n-r_j ]$, where indexing is zero-based and considered modulo $N$ (throughout the paper). While both $x$ and the translations $\{r_j\}$ are  unknown, we stress that the goal here is merely to estimate $x$. This estimation is possible  only up to an arbitrary translation. 

A chief motivation  for this work arises from the imaging technique called single particle Cryo-Electron Microscopy (Cryo-EM), which allows  to visualize molecules at near-atomic resolution \cite{bartesaghi20152,sirohi20163}. In  Cryo-EM, we aim to estimate a three dimensional (3D) object from its two-dimensional (2D) noisy projections, taken at unknown viewing directions~\cite{frank2006three,van2000single}. {While typically the  recovery process involves alignment of multiple observations in a low signal-to-noise ratio (SNR) regime, the underlying goal is merely to estimate the 3D object. In this manner, with the unknown shifts corresponding to the unknown viewing directions, MRA can be understood as a simplified model for Cryo-EM. }
%

Existing approaches for MRA can be classified into two main categories. The first class of methods aims to estimate the set of translations $\{r_j\}$ first. Given this set, estimating $x$ can be achieved easily by aligning all observations $\xi_j$ and then averaging to reduce the noise. The second class, which we favor in this paper, consists of methods which aim to estimate the signal directly, without estimating the shifts.

Considering the first class, one intuitive approach to estimating the translations is to fix a template observation, say $\xi_1$, and to estimate the relative translations by cross-correlation. This is called template matching. Specifically, $r_j$ is estimated as
\begin{equation*}
\hat{r}_j=\arg\max_k\Re\left\{\sum_{n=0}^{N-1}\xi_1[n]\overline{\xi_j[n+k]}\right\},\quad j=2,\cdots,M,
\end{equation*}
where $\Re\{z\}$ and $\overline{z}$ denote the real part and the conjugate of a complex number $z$. This approach requires only one pass over the data: for each observation, the best shift can be computed in ${O}(N\log N)$, and the aligned observations can be averaged online. This results in a total computational cost of ${O}(MN\log N)$, see Table~\ref{tab:AmitsTable}.
While this approach is simple and efficient, it necessarily fails below a critical SNR---see Figure \ref{fig:Example} for a representative example.

The issue with template matching is that we rely on aligning each observation to only one template: this is error prone at low SNR. Instead, to derive a more robust estimator, one can look for the most suitable alignment among all pairs of observations. The $M^2$ relative shifts thus computed must then be reconciled into a compatible choice of $M$ shifts for the individual observations. This is a discrete version of the \emph{angular synchronization} problem, see \cite{singer2011angular,boumal2016nonconvex,perry2016message,chen2016projected,bandeira2014tightness,zhong2017near}. 
The computational complexity of aligning all pairs individually is ${O}(M^2 N\log N)$, while storing the results uses ${O}(M^2)$ memory.  
%

Alternative algorithms for estimating the translations are based on different SDP relaxations \cite{bandeira2014multireference,chen2014near}, iterative template alignment \cite{kosir1995multiple}, zero phase representations \cite{zwart2003fast} and  neural networks \cite{gil2005using}.  The statistical limits of alignment tasks were derived for a variety of setups and noise models, see for instance \cite{aguerrebere2016fundamental,robinson2004fundamental,weiss1983fundamental,weinstein1984fundamental}.
For example, for a continuous, 2D version of the MRA model, it was shown that the Cram\'er--Rao lower bound (CRLB) for translation estimation is proportional to the noise variance $\sigma^2$~\cite{aguerrebere2016fundamental}; crucially, it does not improve with $M${, even if the underlying signal is known.
This is motivation to consider the second category of MRA methods, where shifts are not estimated.} 

Section~\ref{sec:EM} elaborates on expectation maximization (EM) which tries to compute the  maximum marginalized likelihood estimator (MMLE) of the signal---marginalization is done over the shifts. This method acknowledges the difficulty of alignment by working not with estimates of the shifts themselves, but rather with estimates of the probability distributions of the shifts. As a result, EM achieves excellent numerical performance in practice. However its computational complexity is high and its performance is not understood in theory. 
%

\begin{table*} 
\begin{center}
	\begin{tabular}{| l | l | l | l| }
		\hline
		Method & Computational complexity & Storage requirement & Comments \\ \hline
		Template alignment & ${O}(MN\log N)$ & ${O}(N)$ & Fails at moderate SNR (see Figure~\ref{fig:Example}) \\ \hline
		Angular synchronization & ${O}(M^2N\log N)$ & ${O}(M^2)$ & Fails at low SNR \\ \hline
		Expectation maximization & ${O}(T MN\log N )$ & ${O}(MN)$ & Empirically accurate; $\#$iterations $T$ grows with noise level \\ \hline
		Invariant features  (this paper) & ${O}(MN^2 + F(N))$ & ${O}(N^2)$ & Under mild conditions, accurate estimation if $M$ grows as $\sigma^6$ \\
\hline
	\end{tabular}
\caption{Comparison of main MRA approaches. $F(N)$ denotes the complexity of inverting the bispectrum. For instance, for the FM algorithm, $F(N)=O(N^2)$.
	Storage requirements include the possibility of streaming computations where possible. } \label{tab:AmitsTable}
\end{center}
\end{table*}

\begin{figure*}

\centering

\includegraphics[scale=1]{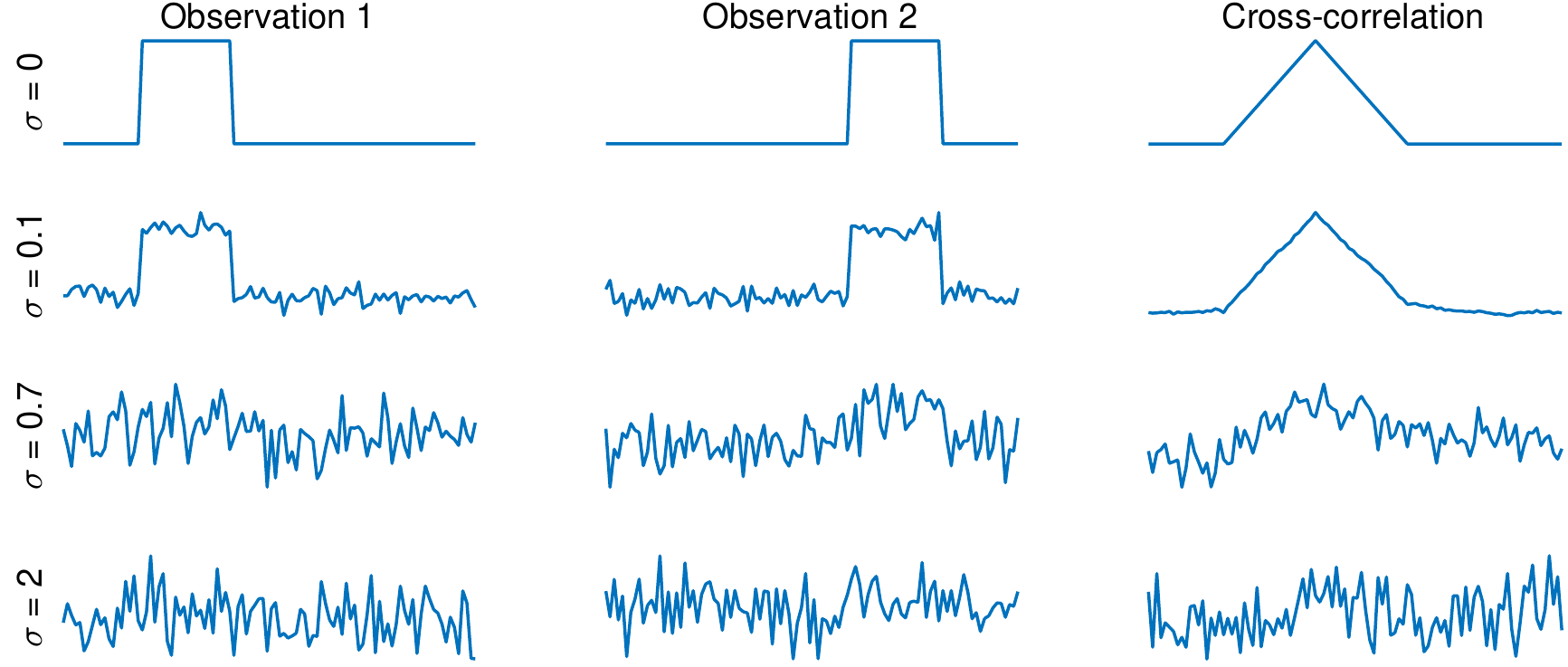}

\protect\caption{\label{fig:Example} Alignment of two translated versions of the same signal  in the presence of  i.i.d. Gaussian noise with various standard deviations $\sigma$. The true signal in $\mathbb{R}^{100}$ is a window of length 22 and height 1. Each  row presents two observations and their cross-correlation. {Importantly, beyond a certain threshold, noise makes pairwise alignment impossible.} }
\end{figure*}

{It has been shown recently that the sample complexity of MRA, under assumption that shifts are distributed uniformly, is proportional to $\sigma^6$ in the low SNR regime. In other words, the number of measurements $M$ needs to scale like $\sigma^6$ to retain a constant estimation error~\cite{bandeira2017optimal}. 

In this work, we propose a framework which achieves this sample complexity by estimating} the sought signal $x$ directly using features that are invariant under translations. For instance, the mean of $x$ is invariant under translation and can be estimated easily from the mean of all observations. 
{We further use the power spectrum and the bispectrum of the observations---which are Fourier-transform based invariants---to estimate the magnitudes and phases of the signal's Fourier transform, respectively.
	
	{For any fixed noise level (which may be arbitrarily large), these features can be estimated accurately provided sufficiently many measurements are available. Hence, 
our approach allows to deal with \emph{any noise level}. }
{Besides achieving the sample complexity, the computational complexity and memory requirements of the methods we describe} are relatively low. Indeed, the only operations whose computational cost grows with $M$ are computations of averages over the data. These can be performed on-the-fly and are easily parallelizable. {We mention that a recent tensor decomposition algorithm also achieves this estimation rate~\cite{perry2017sample}.} 
}

{Given estimators for the mean and power spectrum of $x$, estimating the DC component and Fourier magnitudes of $x$ is straightforward. In this paper, we thus focus on the task of recovering the Fourier phases of $x$ from an estimator of its bispectrum. We propose two {non-convex optimization algorithms on the manifold of phases} for this task, which we call bispectrum inversion. We also discuss three additional algorithms which do not require initialization (and hence could be used to initialize others), based on frequency marching, SDP relaxation and integer programming. The first two methods recover the phases exactly in the absence of noise.}

{Beyond MRA, the bispectrum  plays a central role in a variety of signal processing applications.} For instance, it is a key tool to separate Gaussian and non-Gaussian processes  \cite{brockett1988bispectral,bartolo2004non}. It is also used to investigate the cosmic background radiation \cite{luo1993angular,wang2000cosmic}, seismic signal processing \cite{matsuoka1984phase},  image deblurring \cite{chang1991blur}, feature extraction for radar \cite{chen2008feature}, analysis of EEG signals \cite{ning1989bispectral}, MIMO systems~\cite{chen2001frequency} and classification \cite{zhao2014rotationally} (see also \cite{mendel1991tutorial,nikias1987bispectrum,marabini1994practical,marabini1996new,petropulu1998phase} and references therein). 
{In Section~\ref{sec:proposed_algorithm}, we review previous works on bispectrum inversion~\cite{giannakis1989signal,sadler1992shift,matsuoka1984phase}}.  
Reliable algorithms to invert the bispectrum, as studied here, may prove useful in some of these applications.

The paper is organized as follows.
 Section \ref{sec:invariant_features} discusses the invariant feature approach for MRA. Section \ref{sec:proposed_algorithm} presents {the non-convex algorithms on the manifold of phases for bispectrum inversion.} 
 Section \ref{sec:algorithms} is devoted to additional algorithms that can be used to initialize the non-convex algorithms. Section \ref{sec:optimization_over_phases} analyzes one of the proposed non-convex algorithm.
 Section \ref{sec:EM} elaborates on the EM approach for MRA, Section \ref{sec:numerical_exp} shows numerical experiments and Section \ref{sec:conclusion} offers conclusions and perspective.

Throughout the paper we use the following notation. Vectors {$x$ in $\RN$ or $\CN$ and $y\in\CN$} denote the underlying signal and its discrete Fourier transform (DFT), respectively. {In the sequel, all indices are understood modulo $N$, namely, in the range $0,\ldots,N-1$.}
{The phase of a complex scalar $a$, defined as $a/ \vert a \vert $ if $a\neq 0$ and zero otherwise, is denoted by $\sign(a)$ or $\tilde{a}$}. 
The conjugate-transpose of a vector $z$ is denoted by $z^*$.
We use $'\circ'$ to denote the Hadamard (entry-wise) product, $\mathbb{E}$ for expectation, Tr$(Z)$ for the trace and $\|Z\|_{\mathrm{F}}$ for the Frobenius norm of a matrix $Z$. 
We reserve $T(z)$ for circulant matrices determined by their first row $z$, i.e., $T(z)[k_1,k_2]=z[k_2-k_1]$, and $\mathcal{H}^{N}$ for the set of Hermitian matrices of size $N\times N$.


\section{ Multireference Alignment via Invariant Features }  \label{sec:invariant_features}

We propose to solve the MRA problem directly using features that are {invariant under translation}s. Unlike pairwise alignment, this approach fuses information from all $M$ observations together---not just of pairs---and it only aims to recover the signal itself---not the translations. {The essence of this idea was discussed as a possible extension  in~\cite[Appendix A]{bandeira2014multireference}}. {The invariant features can be understood either as auto-correlation functions or as their Fourier transform. { In this work, we make use of the first three invariants defined as 
	\begin{align}\label{eq:ac}
	c_1 &= \mu_x = \frac{1}{N}\sum_{n=0}^{N-1} x[n], \nonumber \\
	c_2[n_1] &= \frac{1}{N}\sum_{n=0}^{N-1} x[n]\overline{x[n-n_1]}, \\
	c_3[n_1,n_2] &= \frac{1}{N}\sum_{n=0}^{N-1} x[n]\overline{x[n-n_1]}x[n+n_2] \nonumber,
	\end{align}			
{for $n_1,n_2=0,\ldots,N-1$.}
It is clear that $c_1,c_2,c_3$ are invariant under circular shifts of $x$. For higher-order invariants based on auto-correlations, see for instance~\cite{swami1990linear}. 
	}

The first feature is the mean of the signal which is the auto-corrleation function of order one (i.e., $c_1$ in~\eqref{eq:ac}).  The distribution of the mean of $\xi_j$ is then given by  $\mu_{\xi_j}\sim \mathcal{N}\left(\mu_x,\frac{\sigma^2}{N}\right)$ and  we can estimate $\mu_x$ as
\begin{equation} \label{eq:est_mean_x}
\hat{\mu}_x = \frac{1}{M}\sum_{j=1}^M \left(\frac{1}{N}\sum_{n=0}^{N-1}\xi_j[n]\right)\sim\mathcal{N}\left({\mu}_x,\frac{\sigma^2}{NM}\right).
\end{equation}  

Estimating the signal's mean supplies only  limited information about the signal itself. Thus, we consider also the auto-correlation function of order two (i.e., $c_2$ in~\eqref{eq:ac}). Its Fourier transform, the power spectrum, is explicitly defined as
\begin{align*}
	P_x[k] & = \vert y[k]\vert^2, 
\end{align*}
for all $k$, where $y$ is the DFT of $x$. An alternative way to understand the invariance of the power spectrum under shifts is through the effect of shifts on the DFT of a signal:
\begin{align}
	\DFT{R_s x}[k] & = y[k] \cdot e^{-2\pi i k s / N}. 
	\label{eq:shiftDFT}
\end{align}
Thus, shifts only affect the phases of the DFT, so that $P_{R_s x}=P_{x}$ for any shift $R_s$. Furthermore, owing to independence of the noise with respect to the signal itself and to the shift,
\begin{align*}
	\mathbb{E}\left\{P_{\xi_j} [k]\right\} & = P_x[k] + N\sigma^2, 
\end{align*}
where  the second term is the power spectrum of the noise $\varepsilon_j$. Therefore, we estimate the power spectrum of $x$ as:
\begin{align}
	{\hat{P}_x}[k] & = \frac{1}{M}\sum_{j=1}^M(P_{\xi_j}[k]-N\sigma^2).
	\label{eq:ps_estimator}
\end{align} 
It can be shown that ${\hat{P}_x}$ is unbiased and its variance is dominated by $\frac{\sigma^4}{M}$ {for large $\sigma$}. Hence, $\hat{P}_x\rightarrow P_x$ as $M\rightarrow\infty$. In particular, accurate estimation of the power spectrum requires $M$ to scale like $\sigma^4$. In the sequel, we assume that $\sigma$ is known.\footnote{If $\sigma$ is not known, it can be estimated from the data as
\begin{align*}
	\hat{\sigma}^2 & = \frac{1}{N} .\operatorname{variance}\left( \sum_{n=0}^{N-1} \xi_j[n] \right)_{j=1,\ldots,M}.
\end{align*} }

Recovering a signal from its power spectrum  is commonly referred to as \emph{phase retrieval}. This problem received  considerable attention in recent years, see for instance \cite{fienup1982phase,shechtman2015phase,jaganathan2013sparse,bendory2016non,beinert2015ambiguities,bendory2017fourier,bendory2017signal}.
It is well known that almost no one-dimensional signal can be determined uniquely from its power spectrum. Therefore, we use the power spectrum merely to estimate the signal's Fourier magnitudes. As  explained next, we  use the auto-correlation of third order and its Fourier transform, the bispectrum, to estimate the Fourier phases. 


Since phase retrieval is in general ill posed, we use the auto-correlation function of order three (that is, $c_3$ in~\eqref{eq:ac}) through its Fourier transform, the bispectrum, to estimate the Fourier phases of the sought signal. The bispectrum is a function of two frequencies $k_1,k_2=0,\dots,N-1$ and is defined as~\cite{tukey1953}:
\begin{equation} \label{eq:bispec}
B_x[k_1,k_2]=y[k_1]\overline{y[k_2]}y[k_2-k_1].
\end{equation}
Note that, if $y[0]\neq 0 $, the power spectrum is explicitly included in the bispectrum since $P_x[k]={B_x[k,k]}/{y[0]}$. The fact that the bispectrum is invariant under shifts can also be deduced from~\eqref{eq:shiftDFT}. Indeed, for any shift $R_s$,
\begin{align*}
	B_{R_s x}[k_1,k_2] & = \left( y[k_1]e^{-2\pi i k_1 s/N}\right)\left( \overline{y[k_2]}e^{2\pi i k_2 s /N}\right) \\
	& \quad \cdot \left( y[k_2 - k_1]e^{2\pi i (k_1-k_2) s /N}\right) \\
	& = B_{x}[k_1,k_2]. 
\end{align*}
In matrix notation, we express this as
\begin{align}
	B_x = yy^* \circ T(y),
	\label{eq:bispectrum_matrix}
\end{align}
where $T(y)$ is a circulant matrix whose first row is $y$, that is, $T(y)[k_1,k_2]=y[k_2-k_1]$.   Observe that if $x$ is real, then $y[k]=\overline{y[-k]}$ so that $T(y)$ and $B_x$  are Hermitian matrices.
{Simple expectation calculations lead to the conclusion that
	\begin{align}
		\mathbb{E} \left\{ B_{\xi_j}\right\} = B_{x}+ \sigma^2 N^2 \mu_x A,
		\label{eq:bispectrum_estimation}
	\end{align} 
    where $A=A_\mathbb{R}$ or $A=A_\mathbb{C}$ depending on $x\in\RN$ or $x\in\CN$ and
	\begin{equation*}
	\setlength\arraycolsep{4pt}
	A_\mathbb{R} = \begin{bmatrix}
	3 & 1 & 1 &1 & \ldots &1 \\
	1 & 1 & 0 &0 & \ldots &0 \\
	1 & 0 & 1 &0 & \ldots &0 \\
	1 & 0 & 0 &1 & \ldots &0 \\
	\vdots & \vdots & \vdots & \vdots & \ddots &\vdots \\
	1 & 0 & 0 &0  & \ldots &1 \\
	\end{bmatrix}, \ A_\mathbb{C} = \begin{bmatrix}
	2 & 1 & 1 &1 & \ldots &1 \\
	0 & 1 & 0 &0 & \ldots &0 \\
	0 & 0 & 1 &0 & \ldots &0 \\
	0 & 0 & 0 &1 & \ldots &0 \\
	\vdots & \vdots & \vdots & \vdots & \ddots &\vdots \\
	0 & 0 & 0 &0  & \ldots &1 \\
	\end{bmatrix}.
	\end{equation*}
Since the bias term is proportional to $\mu_x$, we propose} to estimate ${\hat{B}_{x-\mu_x}}$  by averaging over $B_{\xi_j-\mu_x}$ for all $j$. This estimator is  unbiased and its variance is controlled by $\frac{\sigma^6}{M}$ {for large $\sigma$}. Therefore, $M$ is required to scale like $\sigma^6$ to ensure accurate estimation. In practice, $\mu_x$ is not known exactly. Thus, we estimate the bispectrum by
\begin{align} 
	{\hat{B}_{x-\mu_x}}=\frac{1}{M}\sum_{j=1}^MB_{\xi_j-\hat{\mu}_x},
	\label{eq:bispec_estimator}
\end{align}
which is asymptotically unbiased. For finite $M$ and large $\sigma$, bias induced by the approximation $\hat{\mu}_x \approx \mu_x$ is significantly smaller than the standard deviation of~\eqref{eq:bispec_estimator}.

The bispectrum contains information about the Fourier phases of $x$ because, defining $\tilde{y}[k]= \sign(y[k])$ and $\tilde{B}_x[k_1,k_2]=\sign({B}_x[k_1,k_2])$ where $\sign$ extracts the phase of a complex number (and returns 0 if that number is 0), we have
\begin{equation} \label{eq:normalized_bispectrum}
	\tilde{B}_x[k_1,k_2] = \tilde{y}[k_1]\overline{\tilde{y}[k_2]}\tilde{y}[k_2-k_1].
\end{equation}
In matrix notation, the normalized bispectrum takes the form $\tilde{B}_x = \tilde{y}\tilde{y}^* \circ T(\tilde{y})$.

Contrary to the power spectrum, the bispectrum is usually invertible. Indeed,  in the absence of noise, the bispectrum determines the sought signal uniquely under moderate conditions:
\begin{proposition} \label{prop:uniqueness}
 For {$N \geq 5$}, let $x\in\CN$ be a signal whose DFT $y$ obeys $y[k]\neq 0 $ for $k= 1, \dots, K$, possibly also for $k = 0$, and zero otherwise.
 Up to integer time shifts, $x$ is determined exactly by its bispectrum provided $K \geq \frac{N+1}{2}$.

For $N \geq 5$, let $x\in\RN$ be a real signal whose DFT $y$ obeys $y[k]\neq 0$ for $k=1, \ldots, K$ and  $k= N-1, \ldots, N-K $, possibly also for $k = 0$, and zero otherwise. Up to integer time shifts, $x$ is determined exactly by its bispectrum provided $\frac{N}{3} \leq K \leq \frac{N-1}{2}$. 
\end{proposition}
\begin{proof}
	This  is a direct corollary of Lemmas \ref{lem:uniquness} and \ref{lem:uniqueness_real}. 
\end{proof}
{We stress that the bispectrum estimator in~\eqref{eq:bispec_estimator} is not a bispectrum itself, since the set of bispectra is not a linear space: $\hat B_{x - \mu_x}$ is not invertible as such~\cite{marabini1994practical}. Algorithms we propose aim to find a \emph{stable} inverse, in the sense that the recovered signal will have a bispectrum which is close to the estimated bispectrum in $\CNN$. The following propositions combined argue formally that this can be done in the MRA model. The proofs in Appendix~\ref{sec:sensitivity} are constructive.}
{
\begin{proposition}[Stable bispectrum inversion] \label{prop:sensitivity}
	There exists an estimator $\hat x$ with the following property. For any signal $x$ in $\RN$ or $\CN$ whose DFT is non-vanishing,
	there exist a precision $\delta = \delta(x) > 0$ and a sensitivity $L = L(x) < \infty$ such that
	if an estimator $\hat B_x$ of $B_x$ satisfies $\|\hat B_x - B_x\|_\mathrm{F} \leq \delta$, then
	$\hat x = \hat x(\hat B_x)$ satisfies $\min_{r = 0\ldots N-1} \|x - R_r \hat x\|_2 \leq L \|\hat B_x - B_x\|_\mathrm{F}$.
\end{proposition}
}
{
\begin{proposition}[Bispectrum estimation] \label{prop:bispectrumestimation}
For any signal $x$ in $\RN$ or $\CN$ whose DFT is non-vanishing, for any required precision $\delta > 0$ and for any probability $p < 1$, there exists a constant $C = C(x, p, \delta) < \infty$ such that, for any noise level $\sigma > 0$, if the number of observations $M$ exceeds $C \cdot (\sigma^2 + \sigma^6)$, the estimator
$$
	\hat B_x = \frac{1}{M} \sum_{j = 1}^M B_{\xi_j} - \sigma^2 N^2 \hat \mu_x A
$$
satisfies $\|\hat B_x - B_x\|_\mathrm{F} \leq \delta$ with probability at least $p$.
\end{proposition}
}


We mention that uniqueness in the continuous setup was considered in \cite{yellott1992uniqueness}. The more general setting of bispectrum over compact groups was considered in \cite{kondor2007novel,kakarala2009completeness,kakarala2012bispectrum,kakarala2009bispectrum}.

The MRA model here assumes  i.i.d.\ Gaussian noise.  However, the estimation is performed by averaging in the bispectrum domain, where noise affecting individual entries is correlated. Consequently, one may want to use a more robust estimator, such as the median. Yet, computing the median of complex matrices is computationally expensive, while computing the average can be performed efficiently and on-the-fly, that is, without requiring to store all observations.
For Gaussian noise, we have noticed numerically that using the mean or the median for bispectrum estimation leads to comparable estimation errors (experiments not shown).
In other noise models, e.g., with outliers, it might be useful to consider the median or the median of means method, see for instance \cite{devroye2016sub}.

Figure \ref{fig:bispectrum_averaging} presents the relative estimation error of the power spectrum and bispectrum as a function of the  number of observations $M$. For the bispectrum, the relative error is computed as
\begin{align*}
	\textrm{relative error} & := \frac{\|{B_{x-\mu_x}}-\frac{1}{M}\sum_{j=1}^MB_{\xi_j-\hat{\mu}_x}\|_{\mathrm{F}}}{\|{B_{x-\mu_x}}\|_\mathrm{F}},
\end{align*}
and similarly for the power spectrum.
 As expected, the slope of all curves is approximately $1/2$ in logarithmic scale, implying that the estimation error decreases as $O(1/\sqrt{M})$. The invariant features approach for MRA is summarized in Algorithm \ref{alg:MRA}.

\begin{algorithm}
	\textbf{Input:} Set of observations $\xi_j,\thinspace j=1,\dots,M$ according to~\eqref{eq:MRA}  and noise level $\sigma$\\
	\textbf{Output:} $\hat{x}$: estimation of $x$ \\
	\textbf{Estimate invariant features:}
	\begin{enumerate}
		\item Compute $\hat{\mu}_x$ according to \eqref{eq:est_mean_x}
		\item Compute $\hat{P}_x$ according to \eqref{eq:ps_estimator}
		\item Compute $\hat{B}_{x-\mu_x}$ according to \eqref{eq:bispec_estimator}				
	\end{enumerate}
	\textbf{Estimate the signal's DFT:}
\begin{enumerate}
	\item Estimate $y[0]$ from $\hat\mu_x$. For other frequencies:
	\item Estimate the magnitudes of $y$ from $\hat{P}_x$
    \item Estimate the phases of $y$ from $\hat{B}_{x-\mu_x}$ (e.g., Algorithm~\ref{alg:GD})
\end{enumerate}
	\textbf{Return:} $\hat{x}$: inverse DFT of the estimated $y$ 
	\protect\caption{\label{alg:MRA} Outline of the invariant approach for MRA}
\end{algorithm}

%
%
%
%
%
\begin{figure*}
	\centering
	\includegraphics[width=\linewidth]{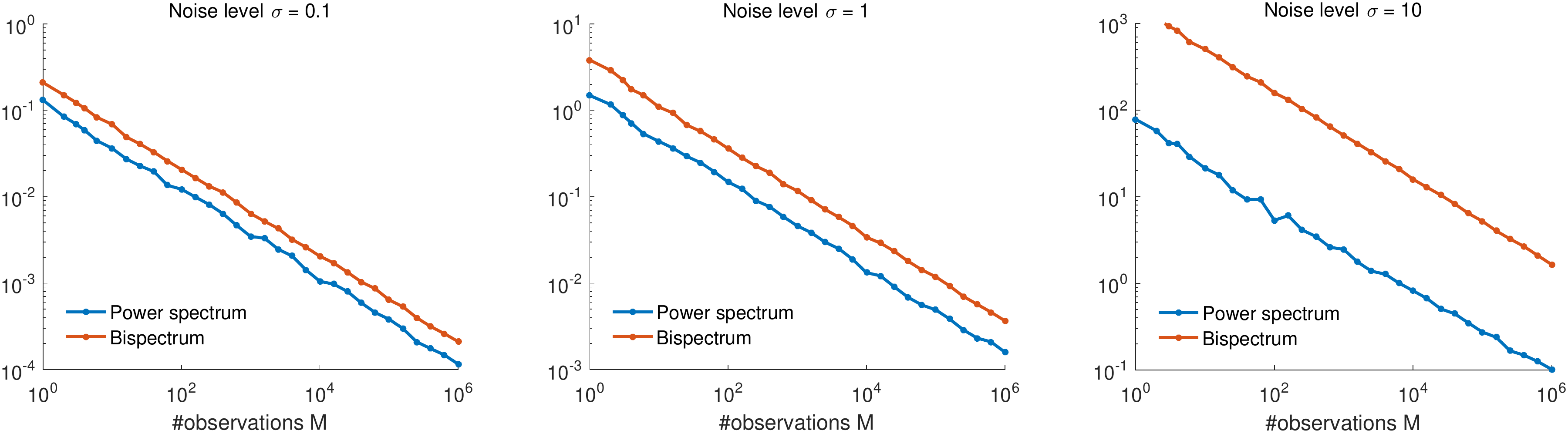}
	\protect\caption{\label{fig:bispectrum_averaging} Relative error of estimating the power spectrum and bispectrum  for different noise levels as a function of the number of observations $M$. Results are averaged over 10 repetitions for each value of $M$ on a fixed real signal of length $N = 41$ with i.i.d.\ normal random entries. {The signal-to-noise ratio is then $1/\sigma^2$.} Importantly, the relative error decreases as $1/\sqrt{M}$ regardless of noise level.}
\end{figure*}

Consider the case in which  the number of samples $M$ may be very large whereas the size of the object is fixed, namely  $N \ll M$. This case is of interest in many applications, such as  cryo-EM~\cite{frank2006three,van2000single}. In this regime, the invariant features approach has two important advantages  over methods that rely on estimating the translations.
First, in the invariant features approach, we  average over the $M$ observations (which is computationally cheap), and then apply a more complex algorithm (say, to recover a signal from its bispectrum) whose input size is a function of $N$ but is independent of $M$. Hence, the overall complexity of this approach can be relatively low. Second, the alignment-based method requires storing all $M$ observations, namely, $MN$ samples, which is unnecessary in the invariant features approach. There, for each observation, we just need to compute its invariants, to be averaged over all observations: this can be done online (in streaming mode) and in parallel.



\section{Non-convex Algorithms for bispectrum inversion} \label{sec:proposed_algorithm}

After estimating the first Fourier coefficient $y[0]$ as $N \hat \mu_x$, our approach for MRA by invariant features consists of two parts. We use the power spectrum to estimate the signal's Fourier magnitudes and the bispectrum for the phases. The first part is straightforward: $|y[k]|$ can be estimated as $\sqrt{\hat P_x[k]}$ if $\hat P_x[k] \geq 0$, and as 0 otherwise. Hereafter, we focus on estimation of the phases of the DFT, $\tilde y$.

In the literature, two main approaches were suggested to invert the discrete bispectrum. The first is based on estimating the frequencies one after the other by exploiting simple algebraic relations \cite{giannakis1989signal,sadler1992shift}. The second approach suggests to estimate the signal by least-squares solution and phase unwrapping \cite{sadler1992shift,matsuoka1984phase}.   
We improve these methods and suggest a few new algorithms. The algorithms are split into two sections. This section is devoted to {two new non-convex algorithms  based on optimization on the manifold of phases}. Both of these algorithms require initialization. While experimentally it appears that random initialization works well, for completeness, in the next section we propose three additional algorithms {which do not need initialization and hence could be used to initialize the non-convex algorithms. }

\subsection{Local non-convex algorithm over the manifold of phases} \label{sec:GD}

In this section, similarly to~\eqref{eq:normalized_bispectrum}, we let $\tilde B$ denote our estimate of the phases of $B_{x}$. Since $\tilde B \approx \tilde y \tilde y^* \circ T(\tilde y)$, one way to model recovery of the Fourier phases $\tilde y$ is by means of the non-convex least-squares optimization problem
\begin{align} 
	\min_{z \in \CN} \left\|W \circ \left(\tilde B - zz^* \circ T(z) \right) \right\|_\mathrm{F}^2 \textrm{ subject to } |z[k]| = 1, \forall k.
	\label{eq:min_fro_manopt}
\end{align}
The matrix $W \in \RNN$ is a weight matrix with nonnegative entries. These weights can be used to indicate our confidence in each entry of $\tilde{B}$.
 Expanding the squared Frobenius norm yields
\begin{align*}
&\left\|W \circ \left(\tilde B - zz^* \circ T(z) \right) \right\|_\mathrm{F}^2 =\|W \circ \tilde B\|_\mathrm{F}^2 \\  & \quad + \left\|W \circ zz^* \circ T(z) \right\|_\mathrm{F}^2  - 2\inner{W \circ \tilde B}{W \circ zz^* \circ T(z)},
\end{align*}
where
\begin{align}
	\inner{U}{V} & = \Re\{\trace(U^*V)\},
	\label{eq:innerproduct}
\end{align}
is the real inner product associated to the Frobenius norm. Under the constraints on $z$, the first two terms are constant and the inner product term is equivalent to
\begin{align*}
	\inner{W \circ \tilde B}{W \circ zz^* \circ T(z)}  = \inner{z}{M(z)z}, 
\end{align*}
 with 
 \begin{align}
 M(z) := W^{(2)} \circ \tilde B \circ \overline{T(z)},	
 \label{eq:Mdef}
 \end{align}	
where we use the notation $W^{(2)} := W \circ W$.
{One possibility is to choose $W = \sqrt{| \hat B_{x-\mu_x} |}$, where the absolute value and the square root are taken entry-wise, so that $M(z) = \hat B_{x-\mu_x} \circ \overline{T(z)}$.}
Hence, optimization problem~\eqref{eq:min_fro_manopt} is equivalent to
\begin{align}
	\max_{z \in \CN} f(z) = \inner{z}{M(z)z} \textrm{ subject to } & |z[k]| = 1, \forall k.	
	\label{eq:maxproblem}
\end{align}
We can also impose $z[0]= \sign(\hat{\mu}_x)$.
 If $x$ is real,  we have the additional symmetry constraints $z[k]=\overline{z[-k]}$.

Since the cost function $f$ is continuous and the search space is compact, a solution exists. 
Of course,
the solution is not unique, in accordance with the invariance of the bispectrum under integer time-shifts of the underlying discrete signal. This is apparent through the fact that the cost function $f$ is invariant under the corresponding (discrete) transformations of $z$. This is true independently of the data $\tilde B$ and $W$. The proof is in Appendix~\ref{sec:proof_lem_manopt_invariance}. 
\begin{lem} \label{lem:manopt_invariance}
	The cost function $f$ is invariant under transformations of $z$ that correspond to integer time-shifts of the underlying signal.
	\end{lem}
 
 To solve this non-convex program, we use the  Riemannian trust-region method (RTR)~\cite{genrtr}, whose usage is simplified by the toolbox Manopt~\cite{manopt}. RTR enjoys global convergence to second-order critical points, that is, points which satisfy first- and second-order necessary optimality conditions~\cite{boumal2016globalrates} and a quadratic local convergence rate.
Empirically, in the noiseless case it appears that the algorithm recovers the target signal with \emph{random initialization}, all local minima are global (with minor technicality for even $N$ in the real case) and all 
second-order critical points have an escape direction, that is, saddles are ``strict''.  Numerical experiments demonstrate reasonable robustness in the face of noise.
This algorithm is summarized in Algorithm \ref{alg:GD} and studied in detail in Section~\ref{sec:optimization_over_phases}. 

\begin{algorithm}
\textbf{Input:} The normalized bispectrum $\tilde{B}[k_1,k_2]$  and a weight matrix $W \in \RNN$ \\
\textbf{Output:} ${\hat{y}}$: an estimation of ${\tilde{y}}$ \\
\textbf{Compute:} Using RTR~\cite{genrtr,manopt}, approximately solve:
\begin{equation*}
\begin{split}
\hat{y}=\arg\max_{z\in\mathbb{C}^N} \Re \left\{z^*M(z)z\right\} \quad\mbox{subject to} \quad  &\vert z[k]\vert =1  ,\thinspace \forall k,\\
\textrm{(if $x$ is real)   }& z[k] = \overline{z[-k]}   ,\thinspace \forall k,
\end{split}
\end{equation*}
where $M(z):=(W\circ W) \circ \tilde B \circ \overline{T(z)}$.
\protect\caption{Non-convex optimization on phase manifold}
\label{alg:GD}
\end{algorithm}


\subsection{Iterative phase synchronization algorithm} \label{sec:APS}

In this section we present an alternative heuristic to the non-convex algorithm on the manifold of phases. This algorithm is based on iteratively solving the phase synchronization problem. 
Suppose we get an estimation of $\tilde{y}$, say $\hat{y}_{k-1}$. If $\hat{y}_{k-1}\approx\tilde{y}$ is non-vanishing, then
this estimation should approximately satisfy the bispectrum relation:
\begin{equation*}
\tilde{B}\circ \overline{T(\hat{y}_{k-1})} \approx \hat{y}_{k-1}^{} \hat{y}_{k-1}^*.
\end{equation*}
The underlying idea is now to push the current estimation towards $\tilde{y}$ by finding a rank-one approximation of  $\tilde{B}\circ \overline{T(\hat{y}_{k-1})}$ with unit modulus entries. This problem can be  formulated as:
\begin{equation} \label{eq:aps}
\arg\max_{z\in\mathbb{C}^N}\Re\left\{z^*\left(\tilde{B}\circ \overline{T(\hat{y}_{k-1})}\right)z\right\} \textrm{   subject to   } \vert z[\ell]\vert=1, \thinspace \forall \ell,
\end{equation}
where we treat the matrix $\tilde{B}\circ \overline{T(\hat{y}_{k-1})}$ as a constant. 
This problem is called phase synchronization. Many algorithms have been suggested to solve the phase synchronization problem. Among them are the eigenvector method, SDP relaxation, projected power method, Riemannian optimization and approximate message passing \cite{singer2011angular,bandeira2014tightness,boumal2016nonconvex,perry2016message,chen2016projected}. Notice that the solution of~\eqref{eq:aps} is only defined up to a global phase, namely, if $z$ is optimal, then so is $ze^{i\phi}$ for any angle $\phi$. To resolve this ambiguity, we require knowledge of the phase of the mean, $\tilde y[0]$ (which is easy to estimate from the data) and we pick the global phase of $\hat y_k$ such that $\hat y_k[0] = \tilde y[0]$.

The $k$th iteration of our algorithm thus (tries to) solve the phase synchronization problem with respect to the matrix $M_{k-1}:=\tilde{B}\circ \overline{T(\hat y_{k-1})}$, where $\hat y_{k-1}$ is the solution of the previous estimation. Assuming the signal is real, we also impose at each iteration the  conjugate-reflection property of $\hat y_k[\ell]=\overline{\hat y_k[-\ell]}$ for all $\ell$ so that $M_{k}$ is Hermitian. 
In the numerical experiments in Section \ref{sec:numerical_exp}, we solve~\eqref{eq:aps} by the Riemannian trust-region method described in~\cite{boumal2016nonconvex}. {Empirically, the performance of this algorithm and Algorithm~\ref{alg:GD} is indistinguishable}. The algorithm is summarized in Algorithm~\ref{alg:APS}.

%

\begin{algorithm}
	\textbf{Input:} The normalized bispectrum $\tilde{B}$, initial estimation $\hat y_0$, phase of the mean $\tilde y[0]$\\
	\textbf{Output:} ${\hat{y}}$: estimation of $\tilde{y}$\\
	Set $k=0$\\
	\textbf{while} stopping criterion does not trigger \textbf{do}:
	\begin{itemize}
		\item[-] $k\leftarrow k+1$
		\item[-] Compute $\hat{y}_k$ as a solution of \eqref{eq:aps}
		\item[-] Fix the global phase: $\hat{y}_k \leftarrow \hat{y}_k \cdot \frac{\tilde y[0]}{\hat y_k[0]}$
		\item[-] If $x$ is real, symmetrize: $\hat{y}_k \leftarrow \sign\left( (\hat{y}_k)_{\downarrow \uparrow}\right)$, see~\eqref{eq:downup}
	\end{itemize}
	\textbf{end while} \\
	\textbf{Return:} $\hat{y} \leftarrow \hat{y}_k$ 
	\protect\caption{Iterative phase synchronization algorithm}
	\label{alg:APS}
\end{algorithm}

\section{ Initialization-free Algorithms } \label{sec:algorithms}

The previous section was devoted to  non-convex algorithms to invert the bispectrum. In this section we present three additional algorithms based on frequency marching (FM), SDP relaxation and phase unwrapping. {These algorithms do not require initialization and therefore could be used to initialize the non-convex algorithms.}

We  prove that FM and the SDP recover the Fourier phases exactly in the absence of noise under the assumption that we can fix $\tilde{y}[1]$. {If the signal has non-vanishing DFT, $\tilde{y}[1]$ can be estimated from the bispectrum using the fact that $\tilde {y}[1]^N$ equals
\begin{align*}
	\sign\left( B_x[N-1, 1] B_x[1, 2] \cdot B_x[1, 2] \cdots  B_x[1, N-1] \right).
\end{align*}
(Any $N$th root can be used for $\tilde{y}[1]$, corresponding to the $N$ possible shifts of $x$.) In all cases, we argue that forcing $\tilde{y}[1] = 1$ is acceptable if $N$ is large. Indeed, recall that a shift by $\ell$ entries in the space domain is equivalent to modulating the $k$th Fourier coefficient by $e^{-2\pi i \ell k/N}$. In particular, it means that
the phase $\tilde y[1]$ can be shifted by $e^{-2\pi i \ell/N}$ for an arbitrary $\ell\in\mathbb{Z}$}. Thus, for signals of length $N \gg 1$, the phase $\tilde y[1]$ can be set arbitrarily with only small error.
In the numerical experiments of Section~\ref{sec:numerical_exp}, we give the correct value of $\tilde y[1]$ to the algorithms in order to assess their best possible behavior.

We begin by discussing the FM algorithm, which is a simple propagation method: it is exact in the absence of noise. Notwithstanding, its estimation for the low-frequency coefficients is sensitive to noise. Because of its recursive nature, error in the low frequencies propagates to the high frequencies, resulting in unreliable estimation. The other two algorithms are more computationally demanding but appear more robust.

\subsection{Frequency marching algorithm} \label{sec:FM}
The FM algorithm is a simple propagation algorithm in the spirit of \cite{giannakis1989signal,sadler1992shift} that aims to estimate  $\tilde{y}$ one frequency at a time. This algorithm has computational complexity $O(N^2)$ and it recovers $\tilde{y}$ exactly for both real and complex signals in the absence of noise, assuming $\tilde{y}[1]$ is known.

Let us denote $\tilde{B}[k_1,k_2]=e^{i\Psi[k_1,k_2]}$ and $\tilde{y}[k]=e^{i\psi[k]}$. Accordingly, we can reformulate (\ref{eq:bispec}) as
\begin{equation*}
\Psi[k_1,k_2]=\psi[k_1]-\psi[k_2]+\psi[k_2-k_1] \bmod 2\pi,
\end{equation*}
where the modulo is taken over the sum of all three terms. 
Using this relation, we can start to estimate the missing phases. The first unknown phase, $\psi[2]$, can  be estimated by:
\begin{equation*}
\begin{split}
\Psi[1,2]&=\psi[1]-\psi[2]+\psi[1] \bmod 2\pi \\ \Rightarrow\quad \hat{\psi}[2]& = 2\psi[1]-\Psi[1,2] \bmod 2\pi,
\end{split}
\end{equation*}
where $\hat{\psi}[2]$ refers to the estimator of $\psi[2]$ (defined modulo $2\pi$). 
We can estimate the next phase in the same manner:
\begin{equation*}
\begin{split}
\Psi[1,3]&=\psi[1]-\psi[3]+\psi[2] \bmod 2\pi\\\Rightarrow\quad \hat{\psi}[3]&={\psi[1]} + \hat{\psi}[2]-\Psi[1,3] \bmod 2\pi.
\end{split}
\end{equation*}

For higher frequencies, we have more measurements to rely on. For the fourth entry, we now can derive two estimators as follows:
\begin{equation*}
\begin{split}
\Psi[1,4]&=\psi[1]-\psi[4]+\psi[3] \bmod 2\pi \\\Rightarrow\quad \hat{\psi}^{(1)}[4]&={\psi[1]} + \hat{\psi}[3]-\Psi[1,4] \bmod 2\pi,
\end{split}
\end{equation*}
and 
\begin{equation*}
\begin{split}
\Psi[2,4]&=\psi[2]-\psi[4]+\psi[2] \bmod 2\pi \\\Rightarrow\quad \hat{\psi}^{(2)}[4]&=2 \hat{\psi}[2]-\Psi[2,4] \bmod 2\pi.
\end{split}
\end{equation*}

In the noiseless case, it is clear that ${\psi}[4]=\hat{\psi}^{(1)}[4]=\hat{\psi}^{(2)}[4]$. In a noisy environment, we can reduce the noise by averaging the two  estimators, where averaging is done over the set of phases (namely, over the rotation group SO(2)) as explained in Appendix \ref{sec:averaging_so2}. Specifically,  
\[
e^{i\hat{\psi}[4]} = \sign\left(e^{i\hat{\psi}^{(1)}[4]}+e^{i\hat{\psi}^{(2)}[4]}\right). 
\]
We can iterate this procedure. To estimate phase $q$, we want to consider all entries of $\tilde B[k, \ell] = \tilde y[k] \overline{\tilde y[\ell]} \tilde y[\ell -k]$ such that exactly one of the indices $k, \ell, $ or $\ell-k$ is equal to $q$ and all other indices are in $1, \ldots, q-1$, so that all other phases involved have already been estimated. A simple verification shows that only entries $\tilde B[p, q],\thinspace p=1,\dots,q-1,$ have that property. Furthermore, because of symmetry in the bispectrum~\eqref{eq:symmetry_prop}, half of these entries are redundant so that only entries $\tilde B[p, q],\thinspace p=1,\dots,\left\lfloor \frac{q}{2}\right\rfloor$ remain. As a result, estimation of the $k$th phase relies on averaging over $\left\lfloor \frac{k}{2}\right\rfloor$ equations, as summarized in Algorithm~\ref{alg:FM}, with the following simple guarantee. {The above construction yields the following proposition.}
\begin{proposition} \label{prop:FM_recovery}
	Let $\tilde{B} = \tilde{B}_x$ be the normalized bispectrum as defined in (\ref{eq:normalized_bispectrum}) and assume  that $\tilde{y}[1]$ {is known}. If $y[k]\neq 0$ for $k=1,\dots,K$, then Algorithm \ref{alg:FM} recovers 
	the Fourier phases $\tilde{y}[k],\thinspace k=1,\dots,K$ exactly.
\end{proposition}
We note in closing that, if the signal $x$ is real, symmetries in the phases $\tilde y$ and $\tilde B$ can be exploited easily in FM.

\begin{algorithm}
	\textbf{Input:} Normalized bispectrum ${\tilde{B}}[k_1,k_2]=e^{i\Psi[k_1,k_2]}$, {$\tilde y[0]$ and $\tilde y[1]\neq 0$}\\
	\textbf{Output:} $\hat{y}$: estimation of $\tilde{y}$ 
	\begin{enumerate}
		\item Set { $\hat{y}[0] = \tilde y[0]$  and $e^{i\hat{\psi}[1]} = \tilde y[1]$ } 
		\item \textbf{For} $k=2,\dots,N$ \textbf{do}: 
		\begin{enumerate}
			\item Average the phase measurements:
			\begin{align*}
				u = \sign\left( \sum_{\ell = 1}^{\left\lfloor\frac{k}{2}\right\rfloor} e^{i\left( \hat{\psi}[\ell]+\hat{\psi}[k-\ell]-\Psi[\ell,k] \right)}  \right)
			\end{align*}
			\item Estimate $\hat{\psi}[k]$ through:
			\begin{equation*} \label{eq:fm_est}
			e^{i\hat{\psi}[k]} = \begin{cases}
			u,& \quad u\neq 0, \\ 1,& \quad u=0.
			\end{cases} 
			\end{equation*}  
		\end{enumerate}
	\end{enumerate}
	\textbf{Return:} $\hat{y} \leftarrow e^{i\hat{\psi}}$
	\protect\caption{\label{alg:FM} Frequency marching algorithm}
\end{algorithm}


\subsection{Semidefinite programming  relaxation} \label{sec:sdp}

In this section we  assume  that the DFT $y$ is non-vanishing so that  the bispectrum relation can be manipulated as 
\begin{align*}
	\tilde{B} & = \tilde{y}\tilde{y}^{*}\circ T(\tilde{y}) & \equiv & & \tilde{B}\circ \overline{T(\tilde{y})} & = \tilde{y}\tilde{y}^{*}, 
\end{align*}
where $\overline{T(\tilde{y})}$ is its entry-wise conjugate.  The developments are easily adapted if the signal has zero mean.
Similarly to the FM algorithm, we assume that $\tilde{y}[0]$ and $\tilde{y}[1]$ are available. 
We aim to estimate $\tilde{y}$ by a convex program. As a first step, we decouple the bispectrum equation and  write the problem of estimating  $\tilde{y}$ as the following non-convex optimization problem:  
\begin{equation} \label{eq:non_convex_sdp}
\begin{split}
\min_{{Z}\in\mathcal{H}^{N},z\in\mathbb{C}^{N}}&\left \|W\circ\left(\tilde{B}\circ\overline{T(z)}-Z\right)\right\|_{\textrm{F}}^2 \\ \mbox{subject to } &  Z= zz^{*},\\& \mbox{diag}\left(Z\right)=1, \\& z[0]=\tilde{y}[0], \thinspace z[1]=\tilde{y}[1],\\ \textrm{(if $x$ is real)\thinspace} &   
z[k]=\overline{z[-k]}, \thinspace \forall k,
\end{split}
\end{equation}
where $\mathcal{H}^{N}$ is the set of Hermitian matrices of size $N$ and $W \in \RNN$ is a real weight matrix with {positive} entries. In particular, in the numerical experiments we set $W=\vert B\vert$.

In the absence of noise, the minimizers of (\ref{eq:non_convex_sdp})  satisfy the bispectrum equation. However, in general these cannot be computed in polynomial time. In order to make the problem tractable, we relax the non-convex coupling constraint $Z = zz^*$ to the convex constraint $Z \succeq zz^*$ (that is, $Z-zz^*$ is positive semidefinite). The convex relaxation is then given by 
\begin{equation} \label{eq:sdp_relax_noisy}
\begin{split}
\min_{Z\in\mathcal{H}^{N},z\in\mathbb{C}^{N}}&\left\|W\circ\left(\tilde{B}\circ\overline{T(z)}-Z\right)\right\|_{\textrm{F}}^2 \\ \mbox{subject to } &   Z\succeq zz^{*},\\& \mbox{diag}\left(Z\right)=1, \\&z[0]=\tilde{y}[0], \thinspace z[1]=\tilde{y}[1],, \\ \textrm{(if $x$ is real)\thinspace} & 
z[k]=\overline{z[-k]}, \thinspace \forall k.
\end{split}
\end{equation}
Upon solving~\eqref{eq:sdp_relax_noisy}, which can be done in polynomial time with interior point methods, the phases $\tilde y$ are estimated from $\sign(z)$. In practice, we use CVX to solve this problem~\cite{grant2008cvx}. The algorithm is summarized in Algorithm \ref{alg:SDP}. We note that problem~\eqref{eq:sdp_relax_noisy} is not a standard SDP, in that its cost function is nonlinear.

\begin{algorithm}
\textbf{Input:} The normalized bispectrum $\tilde{B}$, $\tilde y[0]$ and $\tilde y[1]$  \\
\textbf{Output:} $\hat{y}$: estimation of $\tilde{y}$  \\
\textbf{Solve} the SDP with nonlinear cost function~\eqref{eq:sdp_relax_noisy}, for example using CVX~\cite{grant2008cvx} \\
\textbf{Return:} $\hat{y}\leftarrow \sign(z)$ 

\protect\caption{\label{alg:SDP} Semidefinite relaxation algorithm }
\end{algorithm}

In the noiseless case, the SDP relaxation \eqref{eq:sdp_relax_noisy} recovers the missing phases exactly. Interestingly, the proof is not so much based on optimality conditions as it is on an algebraic property of circulant matrices. The proof of the following property is given in Appendix~\ref{sec:proof_lemma_LA}. 
\begin{lem}\label{lem:LA} Let $\hat{u}$ be the DFT of a vector
$u\in\mathbb{C}^{N}$ obeying $u[k]=\overline{u[-k]}$, so that $\hat{u}$ is real. If $u[0]=u[1]=1$ and $\hat{u}$ is non-negative, 
then $u[k]=1$ for all $k$.
\end{lem}
The following theorem is a direct corollary of Lemma~\ref{lem:LA}. The main proof idea is as follows. Consider $u = \overline{\tilde y} \circ z$ where $(Z, z)$ is optimal for the SDP; then, the constraints ensure $u[0] = u[1] = 1$. Furthermore, one can see via the Schur complement that the constraints force $T(u)$ to be positive semidefinite. Since the eigenvalues of $T(u)$ are the DFT of $u$, it follows that $\hat u$ is non-negative, so that the lemma above applies and $u \equiv 1$, or, equivalently, $z = \tilde y$. Details of the proof are in Appendix~\ref{sec:proof_sdp}.
\begin{thm} \label{th:sdp} For a real signal with non-vanishing DFT $y$, if all weights in $W$ are positive, $\tilde{y}[0]$ and $\tilde{y}[1]$ are known and the objective value of~\eqref{eq:sdp_relax_noisy} attains 0 (which is the case in the absence of noise), then the SDP has a unique solution given by $z=\tilde{y}$ and $Z=zz^*$. 
\end{thm}
We close with an important remark about the symmetry breaking purpose of constraint $z[1] = \tilde y[1]$ in the SDP. Because the signal $x$ can be recovered only up to integer time shifts, even in the noiseless case, without this constraint there are at least $N$ distinct solutions $(z, Z)$ to the SDP. Because SDP is a convex program, any point in the convex hull of these $N$ points is also a solution. Thus, if the symmetry is not broken, the set of solutions contains many irrelevant points. Furthermore, interior point methods tend to converge to a center of the set of solutions, which in this case is never one of the desired solutions.


\subsection{Phase unwrapping by integer programming algorithm} \label{sec:integer_prog}

The next algorithm is based on solving an over-determined system of equations involving integers.
Let us denote $\tilde{y}[k]=e^{i\psi[k]}$ and $\tilde{B}[k_1,k_2]=e^{i\Psi[k_1,k_2]}$ so the normalized bispectrum model is given by
\begin{equation*}
e^{i\Psi[k_1,k_2]}=e^{i(\psi[k_1]-\psi[k_2]+\psi[k_2-k_1])}.
\end{equation*}
By taking the logarithm, we get the algebraic relation
\begin{equation} \label{eq:LLL1}
\Psi[k_1,k_2]+2\pi \chi[k_1,k_2]=\psi[k_1]-\psi[k_2]+\psi[k_2-k_1],
\end{equation}
where, as a result of phase wrapping, $\chi$ takes on integer values. Let $\Psi_{\textrm{vec}}$  and $\chi_{\textrm{vec}}$ be the column-stacked versions of $\Psi$ and $\chi$, respectively. Then, the model reads
\begin{equation} \label{eq:LS}
\Psi_{\textrm{vec}}+2\pi\chi_{\textrm{vec}}=A\psi,
\end{equation}
where the sparse matrix $A\in\mathbb{R}^{N^2\times N}$ encodes the right hand side of~\eqref{eq:LLL1}. It can be verified that $A$ is of rank $N-1$  (see for instance \cite{bendory2017on}), with null space corresponding to the time-shift-induced ambiguity on the phases~\eqref{eq:shiftDFT}.
 Note that both  the integer vector $\chi_{\textrm{vec}}$ and the phases $\psi$ are unknown. Given $\chi_{\textrm{vec}}$, the phases $\psi$ can be obtained easily by solving
\begin{equation} \label{eq:LLL_LS}
\min_{\psi\in\mathbb{R}^N }\| \Psi_{\textrm{vec}}+2\pi \chi_{\textrm{vec}} - A\psi \|_p,
\end{equation} 
for some $\ell_p$ norm. Observe that any error in estimating  $\chi$ may cause a big  estimation error of $\psi$ in \eqref{eq:LLL_LS}. These errors can be thought of as outliers. Hence, we choose to use {least unsquared deviations (LUD)}, $p=1$, which is more robust to outliers. 
The more challenging task is to estimate the integer vector $\chi_{\textrm{vec}}\in\mathbb{Z}^{N^2}$. To this end, we first eliminate  $\psi$ from \eqref{eq:LS} as follows.
Let  $C\in\mathbb{R}^{(N^2-(N-1))\times N^2}$  be a full rank matrix such that $CA=0$, that is, the columns of $C^T$ are in the null space of $A^T$. Matrix $C$ can be designed by at least two methods.
One, suggested in~\cite{marron1990unwrapping}, exploits the special structure of $A$ to design a sparse matrix composed of integer values. Another, which we use here, is to take $C$ to have orthonormal rows which form a basis of the kernel of $A^T$. Numerical experiments (not shown) indicate that the latter approach is more stable. Next, we  multiply  both sides of (\ref{eq:LS}) from the left by $C$ to get 
\begin{equation*} 
C(\Psi_{\textrm{vec}}+2\pi\chi_{\textrm{vec}})=CA\psi=0.
\end{equation*} 
Therefore, the integer recovery problem can  be formulated as 
\begin{equation} \label{eq:LS_int}
\min_{\chi_{\textrm{vec}}\in\mathbb{Z}^{N^2}} \left\| \frac{1}{2\pi} C\Psi_{\textrm{vec}} + C \chi_{\textrm{vec}} \right\|_2,
\end{equation}
where we minimize over all integers.
Note that $C\Psi_{\textrm{vec}}$ is a known vector. The problem is then equivalent to finding a lattice vector
with the basis $C$ which is as close as possible to the vector  $-C\Psi_{\textrm{vec}}/(2\pi)$. While the problem is known to be NP-hard, we approximate the solution of (\ref{eq:LS_int}) with the LLL (Lenstra--Lenstra--Lovasz) algorithm, which can be run in polynomial time~\cite{LLL}. The LLL  algorithm computes a lattice basis, called a \emph{reduced basis},  which is approximately orthogonal. It uses the Gram--Schmidt process to determine the quality of the  basis. For more details, see~\cite[Ch.~17]{galbraith2012mathematics}.

We note that (\ref{eq:LS_int}) is under-determined as the matrix $C$ is of rank $N^2-\mbox{rank}(A)=N^2-(N-1)$. While the LLL algorithm works with under-determined systems, in our case we can solve it for a determined system since we can fix the first $N-1$ entries of $\chi_{\textrm{vec}}$ to be zero.\footnote{We omit the proof of this property here and only mention that it is based on the derivation in \cite{marron1990unwrapping}.}  
Once we have estimated $\chi_{\textrm{vec}}$, we solve \eqref{eq:LLL_LS} with $p=1$. This approach is summarized in Algorithm \ref{alg:integer}.


\begin{algorithm} 
\textbf{Input:} The  normalized bispectrum $\tilde{B}[k_1,k_2]=e^{i\Psi[k_1,k_2]}$\\
\textbf{Output:} $\hat{y}$: estimation of $\tilde{y}$\\ 
\begin{enumerate}
\item (integer programming) Apply the LLL algorithm to estimate the integer vector $\chi_{\textrm{vec}}$ from  
\begin{equation*}
\min_{\chi_{\textrm{vec}}\in\mathbb{Z}^{N^2}}\|C\Psi_{\textrm{vec}}/(2\pi) + C \chi_{\textrm{vec}}\|_2,
\end{equation*}
where $A$ is given in (\ref{eq:LS}), $CA=0$  and $\Psi_{\textrm{vec}}\in\mathbb{R}^{N^2}$ is a column-stacked version of $\Psi$, e.g., using code from~\cite{chang2007miles}.
\item (least-unsquared minimization) Let $\hat{\chi}_{\textrm{vec}}$ be the solution of stage 1. Then, solve 
\begin{equation*}
\hat{\psi}=\arg\min_{\psi\in\mathbb{R}^N}\| \Psi_{\textrm{vec}}+2\pi \hat{\chi}_{\textrm{vec}} - A\psi \|_1.
\end{equation*} 
\end{enumerate}
\textbf{Return:} $\hat{y}\leftarrow e^{i\hat{\psi}}$
\protect\caption{ \label{alg:integer} Phase unwrapping by integer programming}
\end{algorithm}

\section{Analysis of optimization over phases} \label{sec:optimization_over_phases}

In this section, we study the non-convex optimization problem~\eqref{eq:maxproblem} and give more implementation details to solve it, since numerical experiments identify this as the method of choice for MRA from invariant features among all methods compared. We start by considering the general case of a complex signal $x \in \mathbb{C}^N$ and consider the real case in Appendix~\ref{sec:analysis_real}.
 Recall that we aim to maximize 
\begin{equation*}
f(z) = \inner{z}{M(z)z}, \quad M(z):=W^{(2)} \circ \tilde B \circ \overline{T(z)},  
\end{equation*}
where the inner product is defined by~\eqref{eq:innerproduct}, $W$ is a real weighting matrix and $W^{(2)} := W \circ W$. 
The optimization problem lives on a \emph{manifold}, that is, a smooth nonlinear space. Indeed, the smooth cost function $f(z)$ is to be maximized over the set
\begin{equation*} 
	\calM  =  \left\{ z \in \CN : |z[0]| = \cdots = |z[N-1]| = 1 \right\},
\end{equation*}
which is a Cartesian product of $N$ unit circles in the complex plane (a torus). Theory and algorithms for optimization on manifolds can be found in the monograph~\cite{AMS08}. We follow this formalism here. Details can also be found in~\cite{boumal2016nonconvex}, which deals with the similar problem of {phase synchronization}, using similar techniques. For the numerical experiments below, we use the toolbox Manopt which provides implementations of various optimization algorithms on manifolds~\cite{manopt}.

Under mild conditions, the  global optima of~\eqref{eq:maxproblem} correspond exactly to  $\tilde y$ up to  integer time shifts. This fact is proven in Appendix~\ref{sec:proof_uniquness}.
\begin{lem} \label{lem:uniquness}
	For $N \geq 3$, let $x\in\CN$ be a signal whose DFT $y$ is nonzero for frequencies $k$ in $\{1, \ldots, K\}$, possibly also for $k = 0$, and zero otherwise.
	Up to integer time shifts, $x$ is determined exactly by its bispectrum $B$ provided $K \geq \frac{N+1}{2}$. Furthermore, the global optima of~\eqref{eq:maxproblem} correspond exactly to the relevant phases of $y$---up to the effects of integer time shifts---provided $W[k,\ell]$ is  positive  when $B[k,\ell]\neq 0$. 
\end{lem}

The problem at hand is
\begin{align} \label{eq:maxfcalM}
	\max_{z \in \calM} f(z).
\end{align}
This is smooth but non-convex, so that in general it is hard to compute the global optimum. We derive first- and second-order necessary optimality conditions. Points which satisfy these conditions are called \emph{critical} and \emph{second-order critical} points, respectively. Known algorithms converge to critical points (e.g., Riemannian gradient descent) and even to second-order critical points (e.g., Riemannian trust-regions) regardless of initialization~\cite{AMS08,genrtr,boumal2016globalrates}. Empirically, despite non-convexity, the global optimum appears to be computable reliably in favorable noise regimes.

As we proceed to consider optimization algorithms for~\eqref{eq:maxproblem}, the gradient of $f$ will come into play: 
\begin{align*}
	\nabla f(z) & = M(z)z + M(z)^*z + M\adj(zz^*),
\end{align*}
where $M\adj \colon \CNN \to \CN$ is the adjoint of $M$ with respect to the inner product $\inner{\cdot}{\cdot}$. Formally, the adjoint is defined such that, for any $z\in\CN, X \in \CNN$,
\begin{align*}
	\inner{z}{M\adj(X)} & = \inner{M(z)}{X}.
\end{align*}
Specifically, in Appendix \ref{sec:proof_identity_Madj} we show that 
\begin{align}
M\adj(X)[k] & = \trace\left( T_k\transpose \left(W^{(2)} \circ \tilde B \circ \overline{X}\right) \right),
\label{eq:Madjk}
\end{align}
where $T_k$ is a circulant matrix with ones in its $k$th (circular) diagonal and zero otherwise, namely,
\begin{equation} \label{eq:T}
(T_k)[\ell', \ell] =  \begin{cases}
1 & \textrm{ if } \ell' = \ell-k, \\ 0 & \textrm{ otherwise}.
\end{cases}
\end{equation}

As it turns out, under the symmetries of the problem at hand, there is no need to evaluate $M\adj$ explicitly. Indeed, $\tilde B$ obeys
\begin{align} \label{eq:symmetry_prop}
\tilde B[k_2-k_1, k_2] & = \tilde y[k_2-k_1] \overline{\tilde y[k_2]} \tilde y[k_1] = \tilde B[k_1, k_2].
\end{align}
This property is preserved when $\tilde B$ is obtained by averaging bispectra of multiple observations, as in~\eqref{eq:bispec_estimator}. Assuming the same symmetry for the real weights $W$, we find below that $M\adj(zz^*) = M(z)z$. 	See Appendix~\ref{sec:proof_lemma_symmetry_Madj}.
\begin{lem}
	If
	$\tilde B[k_2-k_1,k_2] = \tilde B[k_1, k_2]$ and $W[k_2-k_1,k_2] = W[k_1, k_2]$ for all $k_1, k_2$, then $M\adj(zz^*) = M(z)z$ for all $z\in\CN$. 
	\label{lemma:symmetryMadj}
\end{lem}

Thus, under the symmetries assumed in Lemma~\ref{lemma:symmetryMadj}, the gradient of $f$ simplifies and we get a simple expression for the Hessian as well:
\begin{align} \label{eq:nablaf}
\nabla f(z) & = 2M(z)z + M(z)^*z, \\
\nabla^2 f(z)[\dot z] & = 2M(\dot z)z + 2M(z)\dot z+ M(\dot z)^*z + M(z)^*\dot z  \nonumber. 
\end{align}

For unconstrained optimization, the first-order necessary optimality conditions are $\nabla f(z) = 0$. In the presence of the constraint $z\in\mathcal{M}$, the conditions are different. Namely, following~\cite[eq.~(3.37)]{AMS08}, since $\calM$ is a submanifold of $\CN$, first-order necessary optimality conditions state that the orthogonal projection of the gradient $\nabla f(z)$ to the tangent space to $\calM$ at $z$ must vanish. The result of this projection is called the \emph{Riemannian gradient}. Formally, the tangent space is obtained by linearizing (differentiating) the constraints $|z[k]|^2 = \inner{z[k]}{z[k]} = 1$ for all $k$, yielding
\begin{align*}
	\T_z\calM & = \{ \dot z \in \CN : \inner{z[k]}{\dot z[k]} = 0, \forall k \}.
\end{align*}
Orthogonal projection of $u\in\CN$ to the tangent space $\T_z\calM$ can be computed entry-wise by subtracting from each $u[k]$ its component aligned with $z[k]$. Let $\Proj_z \colon \CN \to \T_z\calM$ denote this projection.
 This operation admits a compact matrix notation as
\begin{align*}
	 u \mapsto \Proj_z(u) & = u - \Re\{ u \circ \overline{z} \}\circ z \\
		& = u - \Re\{\ddiag(uz^*)\}z,
\end{align*}
where $\ddiag \colon \CNN \to \CNN$ sets all non-diagonal entries of a matrix to zero. Equipped with this notion and the expression for $\nabla f(z)$~\eqref{eq:nablaf}, it follows that the Riemannian gradient of $f$ at $z$ on $\calM$ is
\begin{align*}
	\grad f(z) & := \Proj_z(\nabla f(z)) 
	= \nabla f(z) - D(z)z,
\end{align*}
with
\begin{align*}
	D(z) & := \Re\{ \ddiag\left( \nabla f(z)z^* \right) \} = \Re\{ \diag\left(\nabla f(z) \circ \overline{z}\right) \}.
\end{align*}
\begin{lem}
	If $z\in\calM$ is optimal for~\eqref{eq:maxfcalM}, then $\grad f(z) = 0$; equivalently,
	$\diag(\nabla f(z) z^*) = \nabla f(z) \circ \overline{z}$
	is real. 
	\label{lem:critptz}
\end{lem}
\begin{proof}
	See~\cite[Rem.~4.2 and Cor.~4.2]{yang2012optimality}. For the equivalence, notice that $\Proj_z(u) = 0$ if and only if $u[k] = \Re\{u[k] \overline{z[k]}\} z[k]$ for all $k$, and multiply by $\overline{z[k]}$ on both sides using $|z[k]| = 1$.
\end{proof}
A point $z$ which satisfies these conditions is called a critical point.
Likewise, we can define a notion of \emph{Riemannian Hessian} as the linear, self-adjoint operator on $\T_z\calM$ which captures infinitesimal changes in the Riemannian gradient around $z$. Without getting into technical details, we follow~\cite[eq.~(5.15)]{AMS08} and define (with $\D$  the directional derivative operator):
\begin{align*}
	\Hess f(z)[\dot z] & := \Proj_z\left(\D\left( z \mapsto \grad f(z) \right)(z)[\dot z]\right) \nonumber\\
				       & = \Proj_z\left( \nabla^2 f(z)[\dot z] - D(z) \dot z - (\D D(z)[\dot z])z\right),
\end{align*}
where $\D D(z)[\dot z]$ is a real, diagonal matrix. Its contribution to the Hessian is zero, since $(\D D(z)[\dot z]) z$ vanishes under the projection $\Proj_z$. Hence,
\begin{align*}
	\Hess f(z)[\dot z] & = \Proj_z\left( \nabla^2 f(z)[\dot z] - D(z) \dot z \right). 
\end{align*}
The Riemannian Hessian intervenes in the second-order necessary optimality conditions as follows.
\begin{lem} \label{lem:secondordercritz}
	If $z\in\calM$ is optimal for~\eqref{eq:maxfcalM}, then $\grad f(z) = 0$ and $\Hess f(z) \preceq 0$, that is, for all $\dot z \in \T_z\calM$ we have 
	\begin{align*}
		 \inner{\dot z}{\Hess f(z)[\dot z]}  = \inner{\dot z}{\nabla^2 f(z)[\dot z]} - \inner{\dot z}{D(z)\dot z} \leq 0.
	\end{align*}
\end{lem}
\begin{proof}
	See~\cite[Rem.~4.2 and Cor.~4.2]{yang2012optimality}. In the equality, we used the fact that $\Proj_z$ is self-adjoint and $\dot z \in \T_z\calM$.
\end{proof}
A point $z$ which satisfies these conditions is called a \emph{second-order critical point}. With unit weights, the following lemma shows that second-order critical points $z$, in the noiseless case, cannot have an arbitrarily bad objective value $f(z)$. This result is weak, however, since empirically it is observed that in the noiseless case local optimization methods consistently converge to global optima whose value are  $N^2$, suggesting that all second-order critical points are global optima in this simplified scenario. While we do not have a proof for this stronger conjecture, we provide the lemma below because it is analogous to~\cite[Lemma~14]{boumal2016nonconvex} which, in that reference, is a key step toward proving global optimality of second-order critical points.
\begin{lem} \label{lem:second_order_critical_points}
	In the absence of noise and with unit weights, a second-order critical point $z$ of~\eqref{eq:maxfcalM} satisfies $$\nabla f(z) \circ \overline{z} \geq 2(\sqrt{3}-1) > 0.$$ In particular, this implies
	\begin{align*}
		f(z) & = \frac{1}{3} \inner{z}{\nabla f(z)} \geq \frac{2(\sqrt{3}-1)}{3}N. 
	\end{align*}
\end{lem}

\begin{proof}
See Appendix~\ref{sec:proof_lem_second_order_critical_points} for the proof of the inequality. It follows from two key considerations. First, because $z$ is a critical point, Lemma~\ref{lem:critptz} indicates that $\nabla f(z) \circ \overline{z}$ is real. Second, because $z$ is second-order critical, the Riemannian Hessian at $z$ must be negative semidefinite by Lemma~\ref{lem:secondordercritz}. Applied to all tangent directions at $z$ which perturb only one phase at a time implies the desired inequality. The fact that $3f(z) = \inner{z}{\nabla f(z)}$ follows from~\eqref{eq:nablaf}. 
\end{proof}

One final ingredient that is necessary to optimize $f$ over $\calM$ is a means of moving away from a current iterate $z\in\calM$ to the next by following a tangent vector $\dot{z}$. A simple means of achieving this is through a \emph{retraction}~\cite[Def.~4.1.1]{AMS08}. For $\calM$, an obvious retraction is the following:
\begin{align}
	\Retr_z(\dot z) & = \sign(z + \dot z) \in \calM.
	\label{eq:Retr}
\end{align}

With the formalism of~\eqref{eq:maxfcalM} and the above derivations, we can now run a local Riemannian optimization algorithm. As an example, the Riemannian gradient ascent algorithm would iterate the following:
\begin{align*}
	z^{(t+1)} & = \Retr_{z^{(t)}}\left( \eta^{(t)} \grad f(z^{(t)}) \right) \\&= \sign\left( z^{(t)} + \eta^{(t)} \grad f(z^{(t)}) \right),
\end{align*}
where $\eta^{(t)} > 0$ is an appropriately chosen step size and $z^{(0)} \in \calM$ is an initial guess. It is relatively easy to choose the step sizes such that the sequence $z^{(t)}$ converges to critical points regardless of $z^{(0)}$, with a linear local convergence rate~\cite[\S4]{AMS08}. In practice, we prefer to use  the Riemannian trust-region method (RTR)~\cite{genrtr}, whose usage is simplified by the toolbox Manopt~\cite{manopt}. RTR enjoys global convergence to second-order critical points~\cite{boumal2016globalrates} and a quadratic local convergence rate.

In this section, the analysis focused on complex signals. For real signals, we can follow the same methodology while taking the symmetry in the Fourier domain into account. This analysis is given in Appendix \ref{sec:analysis_real}.

\section{Expectation maximization} \label{sec:EM}

In this section, we detail the expectation maximization algorithm (EM)~\cite{dempster1977maximum} applied to MRA. As the numerical experiments in Section~\ref{sec:numerical_exp} demonstrate, EM achieves excellent accuracy in estimating the signal. However, compared to the invariant features approach proposed in this paper, it is significantly slower and requires many passes over the data (thus excluding online processing).

Let $X = \left[\xi_1, \ldots, \xi_M \right]$ be the data matrix of size $N \times M$, following the MRA model~\eqref{eq:MRA}. The  maximum marginalized likelihood estimator (MMLE) for the signal $x$ given $X$ is the maximizer of the likelihood function $L(x ; X) = p(X | x)$ (the probability density of $X$ given $x$). This density could in principle be evaluated by marginalizing the joint distribution $p(X, r | x)$ over the unknown shifts $r \in \{0, \ldots, N-1\}^M$. This, however, is intractable as it involves summing over $N^M$ terms.

Alternatively, EM tries to estimate the MMLE as follows. Given a current estimate for the signal $x_k$, consider the expected value of the log-likelihood function, with respect to the conditional distribution of $r$ given $X$ and $x_k$:
\begin{align}
	Q(x | x_k) & = \mathbb{E}_{r | X, x_k} \! \left\{ \log p(X, r | x_k) \right\}.
\end{align}
This step is called the E-step. Then, iterate by {computing the M-step:} 
\begin{align}
	x_{k+1} & = \arg \max_{x} Q(x | x_k).
\end{align}
For the MRA model, this can be done in closed form. Indeed, the log-likelihood function follows from the i.i.d.\ Gaussian noise model:
\begin{align}
	\log p(X, r | x) & = -\frac{1}{2\sigma^2} \sum_{j = 1}^{M} \|R_{r_j}x - \xi_j\|_2^2 + \mathrm{constant}.
\end{align}
To take the expectation with respect to $r$, we need to compute $w_{k}^{\ell,j}$: for each observation $j$, this is the probability that the shift $r_j$ is equal to $\ell$, given $X$ and assuming $x = x_k$. This also follows easily from the i.i.d.\ Gaussian noise model:
\begin{align}
	w_{k}^{\ell,j} & \propto \exp\left( -\frac{1}{2\sigma^2} \| R_\ell x_k - \xi_j \|_2^2 \right),
\end{align}
(with appropriate scale so that $\sum_{\ell = 0}^{N-1} w_{k}^{\ell,j} = 1$). This allows to write $Q$ down explicitly:
\begin{align*}
	Q(x | x_k) & = -\frac{1}{2\sigma^2} \sum_{j = 1}^M \sum_{\ell = 0}^{N-1} w_{k}^{\ell,j} \|R_\ell x - \xi_j\|_2^2 + \mathrm{constant}.
\end{align*}
This is a convex quadratic expression in $x$ with maximizer
\begin{align}
	x_{k+1} & = \frac{1}{M} \sum_{j = 1}^M \sum_{\ell = 0}^{N-1} w_{k}^{\ell,j} R_\ell^{-1} \xi_j.
\end{align}
In words: given an estimator $x_k$, the next estimator is obtained by averaging all shifted versions of all observations, weighted by the empirical probabilities of the shifts. Considering all shifts of all observations would, in principle, induce an iteration complexity of $O(MN^2)$, but fortunately, for each observation, the matrix of its shifted versions is circulant, which makes it possible to use FFT to reduce the overall computational cost to $O(MN\log N)$. See the available code for details. We note that Matlab naturally parallelizes the computations over $M$.

In practice, we set $x_0 \sim \mathcal{N}(0, I_N)$ to be a random guess. Furthermore, for $M \geq 3000$, we first execute 3000 batch iterations, where the EM update is computed based on a random sample of 1000 observations (fresh sample at each iteration). This inexpensively transforms the random initialization into a ballpark estimate of the signal. The algorithm then proceeds with full-data iterations until the relative change between two consecutive estimates  drops below $10^{-5}$ (in $\ell_2$-norm, up to shifts).

\section{Numerical Experiments } \label{sec:numerical_exp}

This section is devoted to numerical experiments, examining all proposed algorithms. Code for all algorithms and to reproduce the experiments is available online.\footnote{https://github.com/NicolasBoumal/MRA}  The experiments were conducted as follows. The true signal $x$ of length $N = 41$ is a fixed window of height 1 and width 21. {With this signal, the signal-to-noise ratio is $\frac{\| x\|^2}{\|\varepsilon\|^2}\approx \frac{1}{2\sigma^2}$.}
We generated a set of $M$ shifted noisy versions of $x$ as 
\begin{equation*}
\xi_j=R_{r_j}x+\varepsilon_j,
\end{equation*} 
where each shift was randomly drawn from a uniform distribution over $\{0,\dots,N-1\}$ and $\varepsilon_j\sim\mathcal{N}(0,\sigma^2I)$ for all $j$. The relative recovery error for a single experiment is defined  as 
\begin{equation*}
\textrm{relative error}(x,\hat x)= \min_{s\in\{0,\dots,N-1\}}\frac{\|R_{s}\hat{x}-x\|_2}{\|x\|_2},
\end{equation*}
 where $\hat{x}$ is the estimation of the signal. All results are averaged over 20 repetitions.
While we present here results for a specific signal, alternative signal models (e.g., random signals) showed similar numerical behavior.  


%
The following figures compare the recovery errors for all proposed algorithms, with random initialization for those that need initialization.
The non-convex algorithm on the manifold of phases, Algorithm \ref{alg:GD}, runs the Riemannian trust-region method (RTR)~\cite{genrtr} using the toolbox Manopt~\cite{manopt}. Algorithm~\ref{alg:APS} runs 15 iterations with warm-start using the same toolbox. 
 For the phase unwrapping algorithm, Algorithm \ref{alg:integer}, we use an implementation of LLL available in the MILES package \cite{chang2007miles}. 
The SDP is solved with CVX \cite{grant2008cvx}. The EM algorithm is implemented as explained in Section \ref{sec:EM}. We compared the algorithms with an oracle who knowns the random shifts $r_j$ and therefore simply averages out the Gaussian noise.  Experiments are run on a computer with 30 CPUs available. These CPUs are used to compute the invariants in parallel (with Matlab's \texttt{parfor}), while the EM algorithm benefits from parallelism to run the many thousands of FFTs it requires efficiently (built-in Matlab). The algorithms that need $\tilde{y}[0]$ and $\tilde{y}[1]$ are given the correct values. 

Figures \ref{fig:variableMrmse} and \ref{fig:variableMcpu} present the recovery error and computation time of all algorithms as a function of the number of observations $M$ for  fixed noise level $\sigma=1$. Of course, the oracle who knows the shifts of the observations is unbeatable. 
Algorithms~\ref{alg:GD} and~\ref{alg:APS} outperform all invariant approach methods. 
The inferior performance of the SDP might be explained by the fact that we are minimizing a smooth non-linear objective. This is in contrast to SDPs with linear or piecewise linear objectives which tend to promote ``simple'' (i.e., low rank) solutions~\cite{pataki1998rank}, \cite[Remark~6.2]{boumal2015staircase}.
Additionally, while this is not depicted on the figure, we note that for $\sigma=0$ all methods get exact recovery up to machine precision.
EM outperforms the best invariant features approaches by a factor of 3, at the cost of being significantly slower.  For large $M$, the best invariant features approaches are faster than EM by a factor of 25. {Note, however, that for $M$ up to about 300, EM is faster than the other algorithms.}
For invariant features approaches (aside from the SDP), almost all of the time is spent computing the bispectrum estimator, while inverting the bispectrum is relatively cheap. 

Figures \ref{fig:variablesigmarmse} and \ref{fig:variablesigmacpu} show the recovery error and computation time  as a function of  the noise level $\sigma$ with $M = 10,\!000$ observations. Surprisingly, for high noise level $\sigma\gtrsim 3$, the invariant features algorithms outperform EM. 


\begin{figure}
	\centering
	\includegraphics[width=\linewidth]{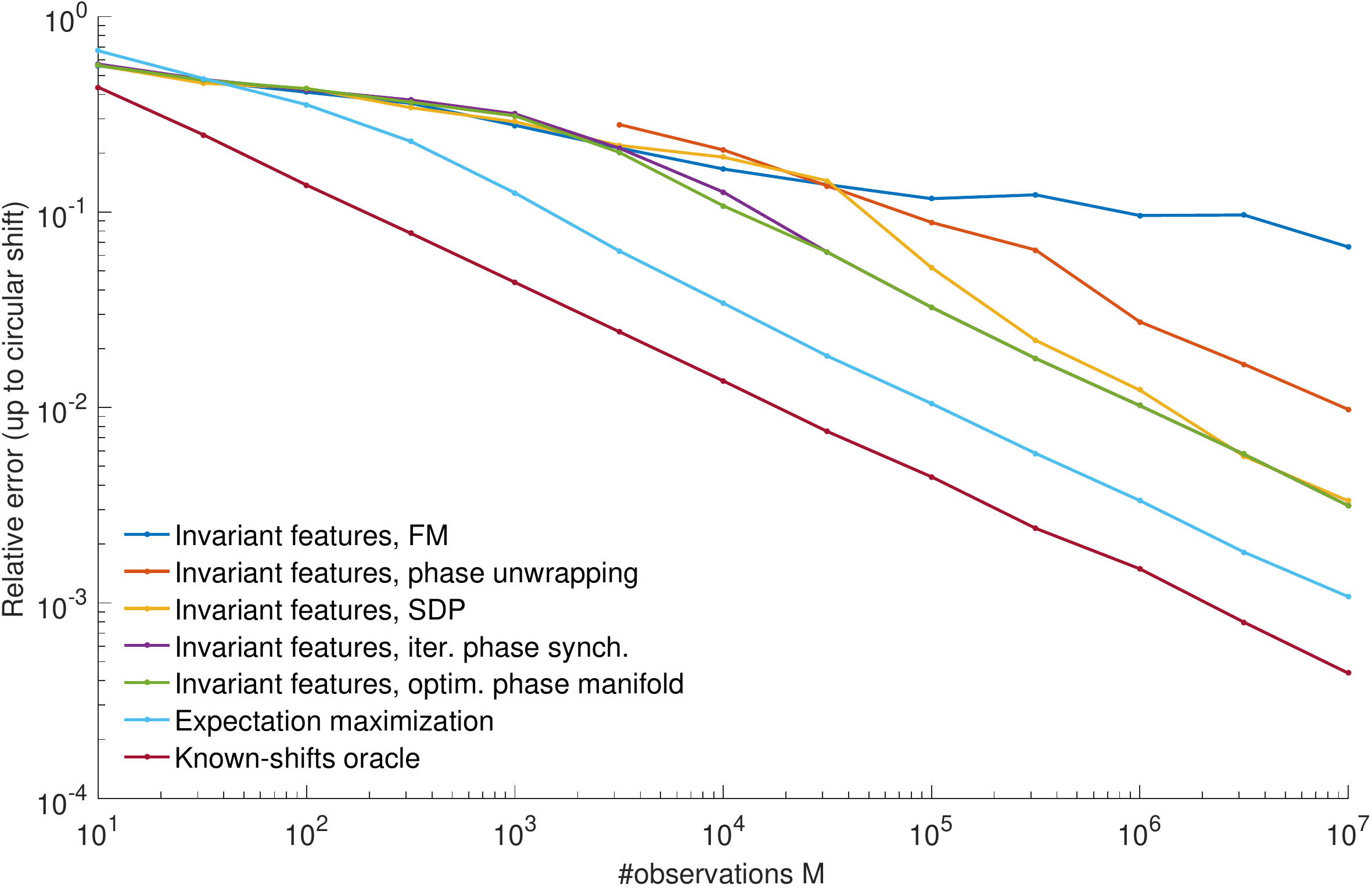}
	\caption{Relative recovery error for the signal $x$ as a function of the number of observations $M$ for fixed noise level $\sigma = 1$.  The curves corresponding to the {optim. phase manifold} (Algorithm~\ref{alg:GD}) and the {iter. phase synch.} (Algorithm~\ref{alg:APS}) overlap.}
	\label{fig:variableMrmse}
\end{figure}

\begin{figure}
	\centering
	\includegraphics[width=\linewidth]{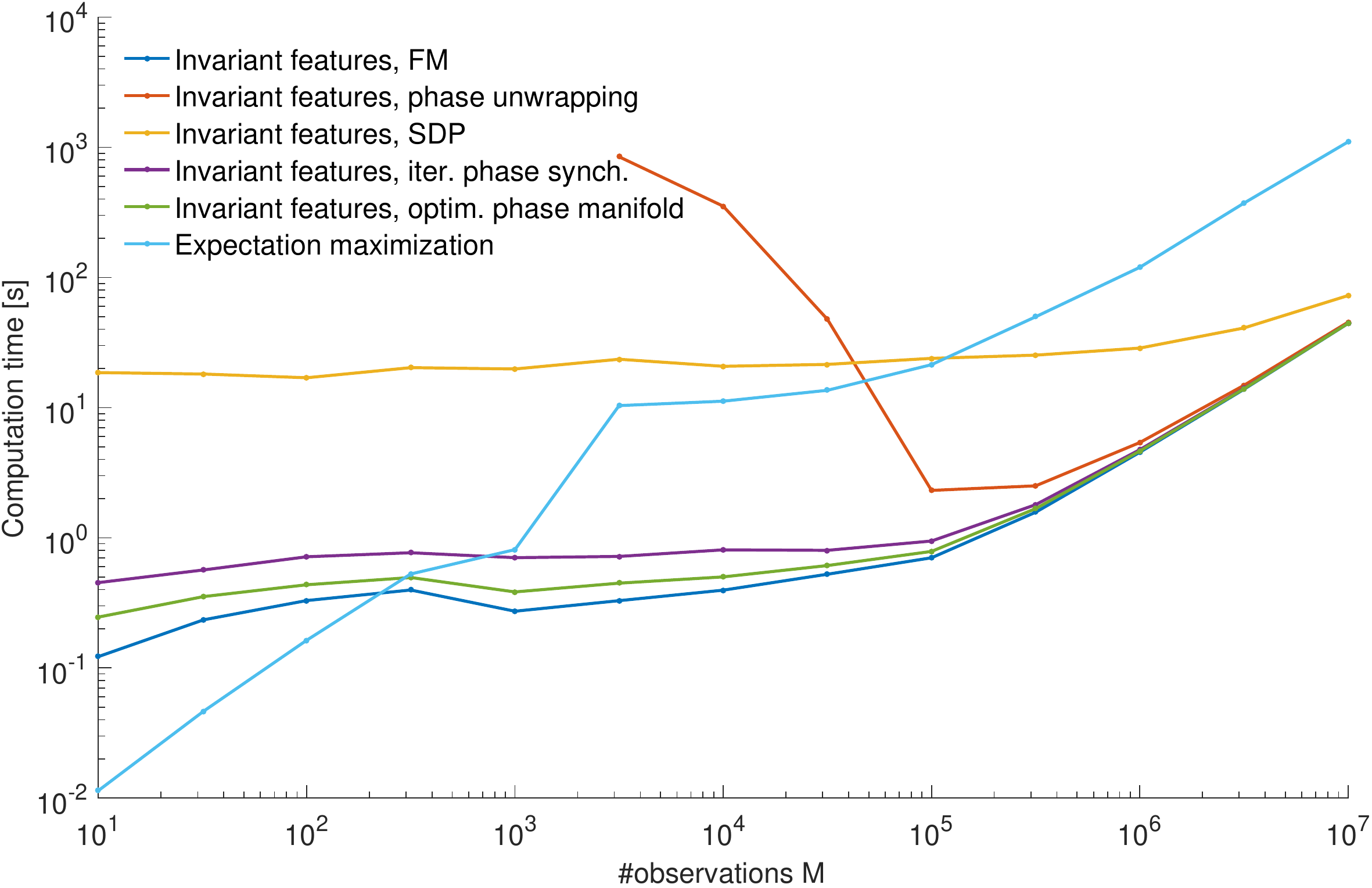}
	\caption{Average computation times corresponding to Figure~\ref{fig:variableMrmse}.}
	\label{fig:variableMcpu}
\end{figure}

\begin{figure}
	\centering
	\includegraphics[width=\linewidth]{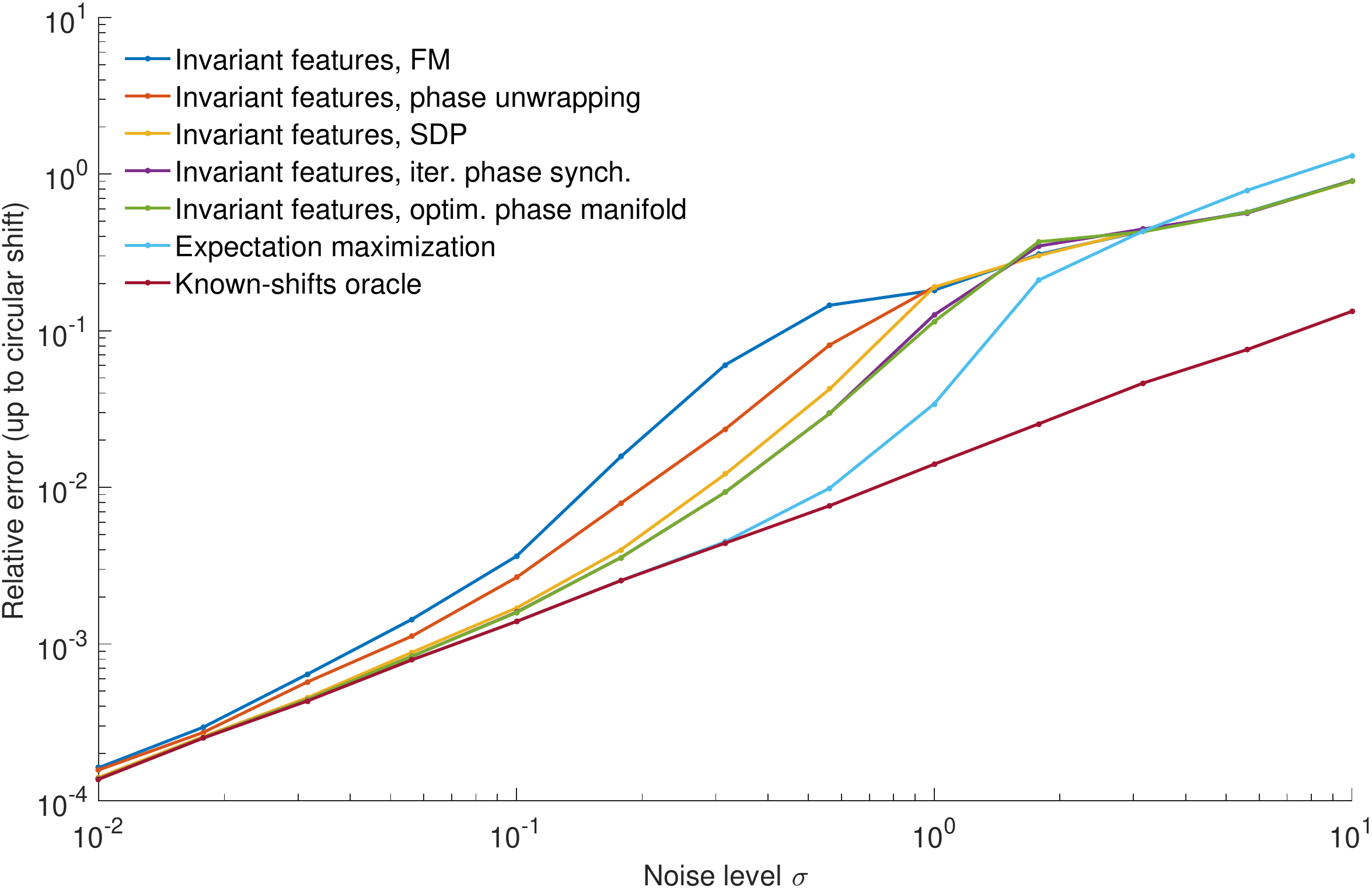}
	\caption{Relative recovery error for the signal $x$ as a function of the noise level $\sigma$  with $M = 10,\!000$ observations. The curves corresponding to the {optim. phase manifold} (Algorithm~\ref{alg:GD}) and the {iter. phase synch.}  (Algorithm~\ref{alg:APS}) overlap.}
	\label{fig:variablesigmarmse}
	\end{figure}

\begin{figure}
	\centering
	\includegraphics[width=\linewidth]{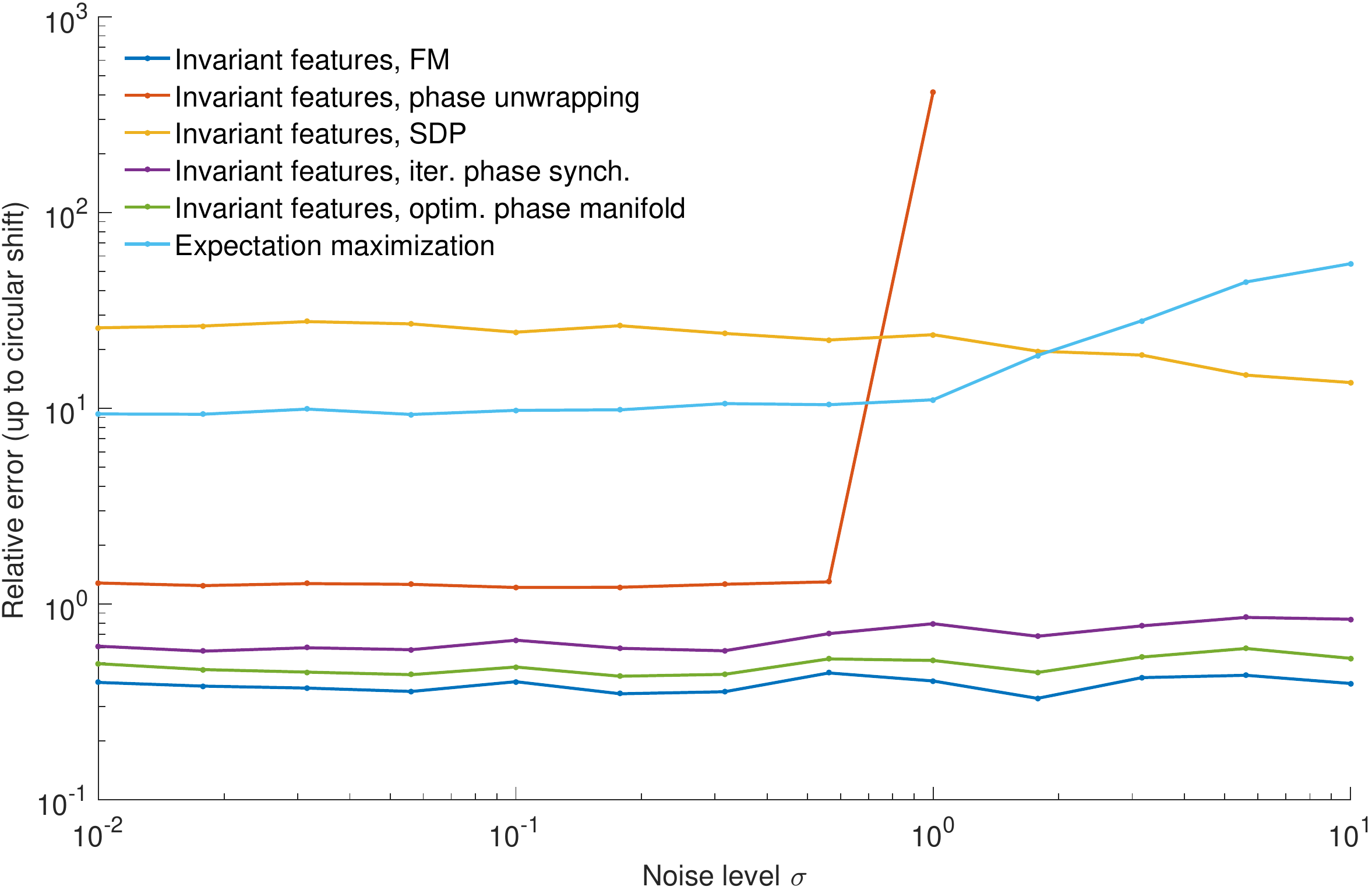}
	\caption{Average computation times corresponding to Figure~\ref{fig:variablesigmarmse}.}
	\label{fig:variablesigmacpu}
\end{figure}

\section{Conclusions and perspective} \label{sec:conclusion}

The goal of this paper is twofold. First, we have suggested a new approach for the MRA problem based on features that are invariant under translations. 
This technique enables us to deal with any noise level as long as we have access to enough measurements and particularly it achieves the sample complexity of MRA. The invariant features approach has low computational complexity and it requires less memory with respect to alternative methods, such as EM.   If one wants to have a highly accurate solution, it can therefore be used to initialize EM.

A main ingredient of the invariant features approach is estimating the signal's Fourier phases by inverting the bispectrum. Hence, the second goal of this paper was to study algorithms for bispectrum inversion.  
 We have proposed a few algorithms for this task. In the presence of noise,  {the non-convex algorithms on the manifold of phases, namely, Algorithms \ref{alg:GD} and \ref{alg:APS}, perform the best}. Empirically, these algorithms have a remarkable property: despite their non-convex landscape, they appear to converge to the target signal from random initialization. We provide some analysis for Algorithm \ref{alg:GD} but this phenomenon is not  well understood.
 
Our chief motivation for this work comes from the more involved problem of cryo-EM. In cryo-EM, a 3D object is estimated from its 2D  projections at unknown rotations in a low SNR environment. One line of research for the object recovery is based on first estimating the unknown rotations  \cite{singer2011three,shkolnisky2012viewing,wang2013orientation}. However, the rotation estimation is performed in a very noisy environment and therefore might be  inaccurate. An interesting question is to examine whether the 3D object can be estimated directly from the acquired data using features that are invariant under the unknown viewing directions~\cite{kam1980reconstruction}.

\section*{Acknowledgement}
{The authors are grateful to Afonso Bandeira, {Roy Lederman, William Leeb, Nir Sharon and Susannah Shoemaker for many insightful discussions. We also thank the reviewers for their useful comments, and particularly the anonymous reviewer who proposed a simpler proof for Lemma~\ref{lem:LA}.}
\
\bibliographystyle{ieeetr}

\bibliography{ref}

\appendix 

\numberwithin{equation}{subsection}
\numberwithin{thm}{subsection} 

\subsection{Stable inversion of bispectrum} \label{sec:sensitivity}
{

We first give a proof of Propositions~\ref{prop:sensitivity}.

Let $U$ denote the set of signals of length $N$ (either real or complex) whose DFTs
are non-vanishing.
Let $G = \{R_0, \ldots, R_{N-1}\}$ denote the discrete group of cyclic shifts over signals of length $N$. In MRA, the parameter space is the quotient space $U/G$, turned into a metric space with distance
\begin{align*}
	\dist([x], [y]) & = \min_{R_r \in G}\|R_r x - y\|_2.
\end{align*}
In the above notation, $[x] = \{R_0x, \ldots, R_{N-1}x\}$ is the equivalence class of signal $x$.
Let $V\subset \CNN$ denote the set of matrices whose entries are nonzero,
equipped with the Frobenius norm distance.

We now construct a function $\psi \colon V \to U/G$.
Crucially, $\psi$ is designed such that if $B\in V$ is the bispectrum of a signal $x$ in $U$, then $\psi(B) = [x]$: it inverts bispectra. Our purpose for building $\psi$ is to establish local Lipschitz continuity.

For a given input $B \in V$, for each $r$ in $\{0, \ldots, N-1\}$, we construct the moduli and phases of $y_r\in\CN$ separately. The inverse DFTs $\{x_0, \ldots, x_{N-1}\}$ form an equivalence class in $U/G$: this is the output $\psi(B)$. Specifically:
\begin{enumerate}
	\item For all $r$, set $|y_r[0]| = \sqrt[3]{|B[0, 0]|}$ and $\sign(y_r[0]) = \sign(B[0, 0])$. 
	\item For all $r$, the moduli $|y_r[k]|$ for $k = 1\ldots N-1$ are:
	\begin{align*}
	|y_r[k]| & = \sqrt{\frac{|B[k, k]|}{|y_r[0]|}}.
	\end{align*}
	\item Assume we are given a reference $\theta_0 \in \reals$ (specified later, independent of input $B$). For a unit-modulus complex number $z$, we define the principal $N$th root of $z$ as $z^{1/N} = e^{i\theta/N}$ with $\theta \in (\theta_0 - \pi, \theta_0+\pi]$ such that $z = e^{i\theta}$.
	\item With this definition of $N$th root, the phase of the first non-DC component $\sign(y_r[1])$ is
	\begin{multline}
	\sign\left(B[N-1, 1] B[1, 2] B[1, 2] \cdots B[1, N-1]\right)^{1/N} \\ \cdot e^{2\pi i r/N}.
	\label{eq:phaseyrone}
	\end{multline}
	(Notice $B[1, 2]$ appears twice.) Each value of $r$ assumes one of $N$ possible $N$th roots.
	\item The phases of $y_r[k]$ for $k = 2\ldots, N-1$ are defined recursively, separately for each $r$:
	\begin{align*}
	\sign(y_r[k]) & = \sign\left( y_r[k-1] \overline{B[1, k]} y_r[1] \right).
	\end{align*}
	\item Define $x_r$ as the inverse DFT of $y_r$; it is easily checked that $x_r = R_rx_0$, so that $\{ x_0, \ldots, x_{N-1} \}$ forms an equivalence class in $U/G$.
\end{enumerate}
We now argue that $\psi$ is locally Lipschitz continuous. 
The construction above implicitly defines a function $x_0 \colon V \to U$, parameterized by a reference $\theta_0$. For a given $B_0 \in V$, pick the reference such that $e^{i\theta_0} = \sign\left(B_0[N-1, 1] B_0[1, 2] B_0[1, 2] \cdots B_0[1, N-1]\right)$. There exists a neighborhood $W$ of $B_0$ such that all $B \in W$ obey
\begin{enumerate}
	\item $|B[k,\ell]| \geq \frac{1}{2}|B_0[k, \ell]|$ for all $k, \ell$, and
	\item $\Re\{ e^{-i\theta_0} \sign\left( Z \right) \} \geq -\frac{1}{2}$, \\ with $Z = B[N-1, 1] B[1, 2] B[1, 2] \cdots B[1, N-1]$.
\end{enumerate}
{In that neighborhood}, $\sign(y_r[1])$ as defined in~\eqref{eq:phaseyrone} is a smooth function of $B$. All other operations involved in computing $\psi(B)$ are smooth in $W$ as well, since entries of $B$ are nonzero. Thus, $x_0 \colon W \to U$ is smooth. In particular, there exists a neighborhood $\bar W \subset W$ of $B_0$ such that $x_0$ is Lipschitz in $\bar W$: there exists $L > 0$ such that
\begin{align*}
\forall B, B' \in \bar W, \quad \|x_0(B) - x_0(B')\|_2 \leq L \|B - B'\|_F.
\end{align*}
Thus, for all $B_0 \in V$ there exists a neighborhood $W$ of $B_0$ in $V$ and $L > 0$ such that
\begin{align*}
\forall B, B' \in W, \quad & \dist(\psi(B), \psi(B'))  \\ & \quad = \min_{r = 0\ldots N-1} \|R_r x_0(B) - x_0(B')\|_2 \\
& \quad \leq \|x_0(B) - x_0(B')\|_2 \\
& \quad \leq L \|B - B'\|_F,
\end{align*}
confirming $\psi$ is locally Lipschitz continuous. In particular, if $B = B_0$ is the bispectrum of $x \in U$ so that $\psi(B) = [x]$, then there exists $\delta > 0$ small enough and $L > 0$ such that if
$\|B' - B\|_\mathrm{F} \leq \delta$, then $[\hat x] = \psi(B')$ obeys
\begin{align*}
\dist([x], [\hat x]) & = \dist(\psi(B), \psi(B')) \leq L \|B - B'\|_F, 
\end{align*}
which shows $x$ can be estimated (up to shift) with finite error from a sufficiently good bispectrum estimator.

We now give a proof of Proposition~\ref{prop:bispectrumestimation}.

{The estimator $\hat B_x$ has expectation equal to $B_x$ and variances on its individual entries are bounded by $c\frac{\sigma^2 + \sigma^6}{M}$ for some $c = c(x)$.}

Chebyshev's inequality states that for a random variable $X$ with mean $\mu$ and variance $\sigma^2$ the probability that $|X-\mu|$ exceeds $t\sigma$ is bounded by $1/t^2$. We have $N^2$ random variables.
By a union bound (which, importantly, does not require independence), we have
\begin{align*}
\operatorname{Prob}\!\left[ \left|\hat B_x[k, \ell] - B_x[k, \ell]\right| \leq t\sqrt{c\frac{\sigma^2 + \sigma^6}{M}} \forall k, \ell \right] > 1 - \frac{N^2}{t^2}.
\end{align*}
Pick $t = \frac{N}{\sqrt{1-p}}$ so that the probability is at least $p$. Now pick $M$ such that $t\sqrt{c\frac{\sigma^2 + \sigma^6}{M}} \leq \delta'$.
Explicitly,
\begin{align*}
M & \geq \frac{cN^2}{(1-p)\delta'^2}(\sigma^2 + \sigma^6).
\end{align*}
This shows that for any $p$ and $\delta'$ we can pick $C$ such that with $M \geq C(\sigma^2 + \sigma ^6)$ we have that $\hat B_x$ approximates $B_x$ within $\delta'$ (entry-wise) with probability at least $p$. Pick $\delta' = \delta/N$ to conclude. (We stress that this expression for $C$ is suboptimal.)


}

\subsection{Analysis of optimization over phases for real signals} \label{sec:analysis_real}

As aforementioned, if the signal under consideration is real, then the phases $\tilde y$ of its Fourier transform obey symmetries. These should be exploited, and the estimator should satisfy the same symmetries. Specifically,
\begin{align*}
\tilde y[N-k] = \overline{\tilde y[k]}, \ \forall k.
\end{align*}
This notably implies that $T(\tilde y)$ is Hermitian, and as a result that $\tilde B$ is Hermitian. It is then also sensible to take $W$ symmetric. Furthermore, this implies that $\tilde y[0]$ is real, so that $\tilde y[0] = \pm 1$, and, if $N$ is even, that $\tilde y[N/2] = \pm 1$ as well. In the latter case, observe that $\tilde y[N/2]$ changes sign when the signal is time-shifted by one index, so that we may fix $\tilde y[N/2] = 1$ without loss of generality, if need be---this is discussed more below. Without loss of generality, let us assume $\tilde y[0] = 1$ (in the MRA framework, this would be estimated from $\hat \mu_x$).

If the signal is real,  the conditions of Lemma \ref{lem:uniquness} can be slightly alleviated so that  
the  global optima of~\eqref{eq:maxproblem} corresponds exactly to  $\tilde y$  as follows:
\begin{lem} \label{lem:uniqueness_real}
	For $N \geq 5$, let $x\in\RN$ be a real signal whose DFT $y$ is nonzero for frequencies $k$ in $\{1, \ldots, K\} \cup \{ N-1, \ldots, N-K \}$, possibly also for $k = 0$, and zero otherwise. Up to integer time shifts, $x$ is determined exactly by its bispectrum $B$ provided $\frac{N}{3} \leq K \leq \frac{N-1}{2}$. Furthermore, the global optima of~\eqref{eq:maxproblem} (with conjugate-reflection constraints) correspond exactly to the relevant phases of $y$---up to the effects of integer time shifts---provided $W[k,\ell]$ is positive when $B[k, \ell]\neq 0$.
\end{lem}
\begin{proof}
	See Appendix \ref{sec:proof_uniqueness_real}.
\end{proof}

In the following, it is helpful to introduce visual notation for the symmetries. Given any vector $u \in \CN$, consider the following linear operations:
\begin{align*}
u_\downarrow & = u = \begin{pmatrix} u[0] \\ u[1] \\ u[2] \\ \vdots \\ u[N-1] \end{pmatrix}, & u_\uparrow & = \begin{pmatrix} u[0] \\ u[N-1] \\ \vdots \\ u[2] \\ u[1] \end{pmatrix},
\end{align*}
and 
\begin{equation} \label{eq:downup}
u_{\downarrow\uparrow} = \frac{u_\downarrow + \overline{u_\uparrow}}{2}.
\end{equation}
Using this notation, the symmetries of $z$, namely, $z[N-k] = \overline{z[k]}$ for all $k$, can be written $z_\uparrow = \overline{z_\downarrow}$. As a result, for odd $N$ the phase estimation problem lives on the following submanifold of $\calM$:
\begin{align*}
\calM_\reals & = \Big\{ z \in \calM : z_\uparrow = \overline{z_\downarrow}, \ z[0] = 1  \Big\}.
\end{align*}
If $N$ is even, the submanifold takes a slightly different form: 
\begin{align*}
\calM_\reals & = \Big\{ z \in \calM : z_\uparrow = \overline{z_\downarrow}, \ z[0] = z[N/2]= 1  \Big\}.
\end{align*}
The corresponding optimization problem is
\begin{align}
\max_{z \in \calM_\reals} f(z) = \inner{z}{M(z)z}.
\label{eq:maxfcalMreal}
\end{align}

To carry out the optimization of~\eqref{eq:maxfcalMreal}, we follow the same protocol as for the complex case, namely, we obtain expressions for the Riemannian gradient and the Riemannian Hessian of the cost function $f$ restricted to the manifold $\calM_\reals$, at which point we will be in a position to use standard algorithms.

The first step is to obtain an expression for the orthogonal projector from $\CN$ to the tangent spaces of $\calM_\reals$. These tangent spaces are readily obtained by differentiating the constraints:
\begin{align*}
\T_z\calM_\reals & = \left\{ \dot z \in \T_z\calM : \dot z_\uparrow = \overline{\dot z_\downarrow} \right\}.
\end{align*}
This implicitly forces $\dot z[0] = 0$ and, if $N$ is even, $\dot z[N/2] = 0$. The orthogonal projection from $\CN$ to that linear space is
\begin{align*}
\Proj_z^\reals(u) & = \Proj_z\left(\frac{u_\downarrow + \overline{u_\uparrow}}{2}\right) = \Proj_z\left(u_{\downarrow\uparrow}\right).
\end{align*}
One can check that, for any $u$, $\left(\Proj_z^\reals u\right)[k] = 0$ if $k = 0$ and if $k = N/2$ for $N$ even, as expected. As for the complex case, these steps yield explicit expressions for the Riemannian gradient and Hessian of $f$ on $\calM_\reals$:
\begin{align*}
\grad f(z) & = \Proj_z^\reals\left( \nabla f(z) \right) \nonumber\\
& = \nabla f(z)_{\downarrow\uparrow} - \Re\{\nabla f(z)_{\downarrow\uparrow} \circ \overline{z}\} \circ z, 
\end{align*}
and
\begin{align*}
\Hess f(z)[\dot z] & = \Proj_z^\reals\left( \D\left(z \mapsto \grad f(z)\right)(z)[\dot z] \right) \nonumber\\
& = \Proj_z^\reals\left( \nabla^2 f(z)[\dot z]_{\downarrow\uparrow} - \Re\{\nabla f(z)_{\downarrow\uparrow} \circ \overline{z} \} \circ \dot z\right) \nonumber \\
& = \Proj_z^\reals\left( \nabla^2 f(z)[\dot z] - \Re\{\nabla f(z)_{\downarrow\uparrow} \circ \overline{z} \} \circ \dot z\right). 
\end{align*}
To reach the last equality, we used the fact that $\Proj_z^\reals(\alpha \circ z) = 0$ for any real vector $\alpha$, so that one of the terms vanished under the projection.

Interestingly, for $z$ restricted to $\calM_\reals$, the cost function and its derivatives simplify. Indeed, $T(z)$ is now Hermitian, so that $M(z)$ is Hermitian (using a symmetric $W$) and~\eqref{eq:nablaf} becomes:
\begin{align*}
&f(z) = z^* M(z) z, \\
&\nabla f(z)  = 3M(z)z,  \\
&\nabla^2 f(z)[\dot z]  = 3M(\dot z)z + 3M(z)\dot z = 6M(z)\dot z.
\end{align*}
The last simplification of the Hessian follows from this result, under symmetry assumptions which are valid in the real case.
\begin{lem} \label{lem:Mz}
	Under the same symmetry conditions as in Lemma~\ref{lemma:symmetryMadj}, if additionally $W$ and $\tilde B$ are Hermitian, then, if $z\in\calM_\reals$ and $\dot z \in \T_z\calM_\reals$, it holds that $M(z)\dot z = M(\dot z)z$.
\end{lem}
\begin{proof}
	See Appendix \ref{sec:proof_lem_Mz}.
\end{proof}

It remains to define a retraction on $\calM_\reals$. An obvious choice is:
\begin{align*}
\Retr_z^\reals(\dot z) & = \Retr_z(\dot z),
\end{align*}
where $\Retr_z$ was defined in~\eqref{eq:Retr}. Indeed, it is a simple exercise to check that $\Retr_z(\dot z)$ is in $\calM_\reals$ for any $z\in\calM_\reals$ and $\dot z \in \T_z\calM_\reals$:
\begin{align*}
\Retr_z(\dot z)_\uparrow & = \sign(z + \dot z)_\uparrow = \sign(z_\uparrow + \dot z_\uparrow) \\& = \sign(\overline{z_\downarrow + \dot z_\downarrow}) = \overline{\Retr_z(\dot z)_\downarrow}.
\end{align*}

Empirically, with exact data $\tilde B$ and unit weights $W$, we find that the parity of $N$ has an important effect on the nature of second-order critical points of~\eqref{eq:maxfcalMreal}. If $N$ is odd and $z[0] = \tilde y[0]$, for $N = 3, 5, 7, 9$, we find empirically that $f \colon \calM_\reals \to \reals$ has $N$ second-order critical points which are all global optima; they correspond to the phases of the DFT of all possible time-shifts of the unknown signal. Hence, any second-order critical point leads to exact recovery. On the other hand, if $N$ is even, we find that only $N/2$ second-order critical points are global optima. If $N = 4, 6, 8$, there are no other second-order critical points. If $N = 10, 12$, there are $N/2$ additional second-order critical points. These are strict, non-global local optima. Interestingly, if one flips the sign of $z[N/2]$, one recovers the phases of one of the other $N/2$ possible time-shifts of the unknown signal. This is why, for even $N$, we recommend the following: (i) compute a second-order critical point $z$ of~\eqref{eq:maxfcalMreal} from some initial guess; (ii) let $z'$ be $z$ with the sign of $z[N/2]$ flipped; (iii) if $f(z') > f(z)$, run the optimization again with $z[N/2]$ flipped in $\calM_\reals$, this time starting at $z'$. In the noiseless case, empirically, $z'$ will already be optimal.

Of course, the parameterization of $\calM_\reals$ as proposed here is redundant. Given the symmetries, one could alternatively choose to work with only the phases $z[1]$ to $z[\lfloor\frac{N-1}{2}\rfloor]$, which are sufficient to determine all of $z\in\calM_\reals$. This is certainly computationally advantageous. The choice to work with a redundant parameterization above is motivated by the simpler exposition it allows. There is no conceptual or technical difficulty in implementing the above with a non-redundant parameterization.

\subsection{Proof of Lemma \ref{lem:manopt_invariance}} \label{sec:proof_lem_manopt_invariance}

Defining $z_\theta$ as $z_\theta[k] = z[k]e^{i\theta g(k)}$ with $g(k) = k \bmod N$, the statement to prove is equivalent to $f(z) = f(z_\theta)$ for any $\theta$ which is an integer multiple of $\frac{2\pi}{N}$.

Let $\omega = e^{i\theta}$ and let $u \in \CN$ be the vector defined by $u[k] = \omega^{g(k)}$. Then, $z_\theta = z \circ u = \diag(u)z$ so that
\begin{align*}
f(z_\theta) = \Re\{ z^* \diag(u)^* M(z_\theta) \diag(u) z \}.
\end{align*}
Using~\eqref{eq:Mdef} for the definition of $M$ we get
\begin{align*}
M(z_\theta) & = W^{(2)} \circ \tilde B \circ \overline{T(z \circ u)} = M(z) \circ \overline{T(u)}.
\end{align*}
Hence,
\begin{align*}
\diag(u)^* M(z_\theta) \diag(u) = M(z) \circ \diag(u)^* \overline{T(u)} \diag(u).
\end{align*}
All dependence on $\theta$ now resides in the matrix $\diag(u)^* \overline{T(u)} \diag(u)$. Its entries obey
\begin{align*}
\left(\diag(u)^* \overline{T(u)} \diag(u)\right)[k_1, k_2] & = \overline{u[k_1]} \overline{u[k_2-k_1]} u[k_2] \\& = \omega^{g(k_2) - g(k_1) - g(k_2-k_1)}. 
\end{align*}
For $(k_1, k_2)$ ranging from 0 to $N-1$, the exponent $g(k_2) - g(k_1) - g(k_2-k_1)$ evaluates to 0 if $k_2 \geq k_1$ and to $N$ otherwise. In the first case, $\omega^0 = 1$. In the second case, $\omega^N = e^{i\theta N} = 1$ owing to the assumption that $\theta$ is an integer multiple of $\frac{2\pi}{N}$. Consequently,
\begin{align*}
\diag(u)^* M(z_\theta) \diag(u) = M(z)
\end{align*}
and indeed $f(z) = f(z_\theta)$.


\subsection{Averaging over phases} \label{sec:averaging_so2}

Let $a_1,a_2,\dots,a_K$ be complex numbers with unit modulus. We define the average over the set of phases (i.e., the group SO(2)) as the solution of
\begin{equation*}
\min_{z\in\mathbb{C}}\sum_{k=1}^K\vert z - a_k\vert^2\quad\mbox{subject to}\quad \vert z\vert =1.
\end{equation*}
Expanding the square modulus, the problem is equivalent to 
\begin{equation*}
\max_{z\in\mathbb{C}}\Re\left\{ \overline{z}\sum_{k=1}^Ka_k \right\}\quad\mbox{subject to}\quad \vert z\vert =1.
\end{equation*}
The last expression is maximized if and only if $\overline{z}\sum_{k=1}^Ka_k$ is real and non-negative, i.e., $z$ is the phase of $\sum_{k=1}^Ka_k$. Explicitly, if $\sum_{k=1}^Ka_k\neq 0$ then  
\begin{equation*}
z = \sign\left(\sum_{k=1}^Ka_k\right).
\end{equation*}
Otherwise, any unit-modulus $z$ is optimal.

\subsection{Proof of Lemma \ref{lem:LA}} \label{sec:proof_lemma_LA}
{By expanding $u[1]$, we get 
\begin{equation*}
\begin{split}
1 &= u[1] = \Re\{u[1]\} = \frac{1}{N} \Re\left\{ \sum_{k=0}^{N-1}\hat{u}[k]e^{2\pi i k/N} \right\} \\ 
& = \frac{1}{N}  \sum_{k=0}^{N-1}\hat{u}[k]\Re\left\{e^{2\pi i k/N} \right\},
\end{split}
\end{equation*}
where the last equality holds since $\hat{u}$ is real. Therefore, since $\hat{u}$ is non-negative, we get
\begin{equation*}
\begin{split}
1 &= u[1] \leq   \frac{1}{N}  \sum_{k=0}^{N-1}\hat{u}[k]   =  u[0] = 1.
\end{split}
\end{equation*}
This equality holds if and only if $\hat{u}[k] =0$ whenever $\Re\left\{e^{2\pi i k/N} \right\}<1$. This implies in turn that $\hat{u}$ is a delta function and thus $u$ is a constant. 
}

\subsection{Proof of Theorem \ref{th:sdp}} \label{sec:proof_sdp}
In the noiseless case, plugging in the correct signal into \eqref{eq:sdp_relax_noisy} satisfies the constraints and attains 0 for the objective value, which is clearly minimal. Thus, any global optimum $(Z, z)$ must have objective value 0, meaning
 $\tilde{B}\circ\overline{T(z)} = Z$.
 
Using Schur's complement, the SDP constraint $Z \succeq zz^*$ can be written as
\begin{equation*}
Q:=\begin{bmatrix}Z & z\\
z^{*} & 1
\end{bmatrix} = \begin{bmatrix}\tilde{B}\circ\overline{T(z)} & z\\
z^{*} & 1
\end{bmatrix} \succeq 0.
\end{equation*} 

Let \[ U=\begin{bmatrix}\mbox{diag}\left(\tilde{y}\right) & 0\\
0 & 1
\end{bmatrix},\]
where $\mbox{diag}\left(\tilde{y}\right)$ is a diagonal matrix with the entries of $\tilde{y}$. Observe that $Q\succeq0$ implies $U^{*}QU\succeq 0$. Recall that  each entry of $\tilde{y}$ has modulus one. Then, using $\tilde B = \tilde y \tilde y^* \circ T(\tilde  y)$, we get 
\begin{align*}
U^{*}QU & =\begin{bmatrix}\mbox{diag}\left(\overline{\tilde{y}}\right) & 0\\
0 & 1
\end{bmatrix}\begin{bmatrix}\tilde B\circ\overline{T(z)} & z\\
z^{*} & 1
\end{bmatrix}\begin{bmatrix}\mbox{diag}\left(\tilde{y}\right) & 0\\
0 & 1
\end{bmatrix}\\
 & =\begin{bmatrix}T(\tilde{y}\circ\overline{z}) & \overline{\tilde{y}}\circ z\\
\left(\overline{\tilde{y}}\circ z\right)^{*} & 1
\end{bmatrix}\\&=\begin{bmatrix}T(\overline{u}) & u\\
u^{*} & 1
\end{bmatrix}\\ &\succeq 0,
\end{align*}
where  $u:=\overline{\tilde{y}}\circ z$. This in turn implies that $T(\overline{u})\succeq 0$, hence, by properties of circulant matrices,  $\hat{u}$ (the DFT of $u$) is non-negative. Furthermore, 
since  $z[0]=\tilde{y}[0]$
and $z[1]=\tilde{y}[1]$, it follows that $u[0]=u[1]=1$. So, by Lemma \ref{lem:LA} we
conclude that $u[k]=1$ for all $k$ and $z = \tilde{y}$. This implies immediately $Z=zz^*=\tilde{y}\tilde{y}^*$.


\subsection{Proof of Lemma \ref{lem:uniquness}} \label{sec:proof_uniquness}

We start by characterizing the global optima of~\eqref{eq:maxproblem}.
Recall that 
\begin{align*}
M(z) & = W^{(2)} \circ \tilde B \circ \overline{T(z)} = W^{(2)} \circ \tilde y\tilde y^* \circ T(\overline{z} \circ \tilde y).
\end{align*}
Defining $u = \overline{z} \circ \tilde y$, the objective function of~\eqref{eq:maxproblem} simplifies to
\begin{align*}
f(z) & = \inner{z}{M(z)z} = \sum_{k,\ell=0}^{N-1} W[k, \ell]^2 \Re\left\{ u[k] \overline{u[\ell]} u[\ell-k]  \right\}.
\end{align*}
Clearly, $f(\tilde y)$ is maximal over $\calM$. Any other global optimum must attain the same objective value. By our assumptions on $W$, this implies for all  ${k, \ell, \ell - k~\in\{1, \ldots, K\} \bmod N}$  that
\begin{align} 
 u[k] \overline{u[\ell]} = \overline{ u[\ell-k] }.
\label{eq:invertibilitymaster}
\end{align}
Consider $k = 1$ and $\ell = 2, \ldots, K$; this easily leads to $u[k] = u[1]^k$ for $k\in\{1, \ldots, K\}$.
Now set $k = K, \ell = 1$. Using periodicity of indexing, since $\ell - k = 1 - K = N + 1 - K \bmod N$, and using the conditions $K \geq \frac{N+1}{2}$ and $K \leq N-1$, we have $N + 1 - K \in \{2, \ldots, K\}$, so that~\eqref{eq:invertibilitymaster} applies. On one hand, we find $u[\ell-k] = u[N+1-K] = u[1]^{N+1-K}$, and on the other hand we find $\overline{u[\ell-k]} = u[k] \overline{u[\ell]} = u[1]^{K-1}$. Thus, $u[1]^N = 1$. If we let $u[1] = e^{i\theta}$ for some $\theta \in \reals$, it follows that $\theta = m\frac{2\pi}{N}$ for some integer $m$. If $x$ has zero mean, $u[0]$ is irrelevant and can be set to $1$; otherwise, equation~\eqref{eq:invertibilitymaster} also holds with $k=\ell=0$, so that $u[0]=1$. Recalling $u = \overline{z} \circ \tilde y$, it follows that if $z$ is optimal, then $z[k] = \tilde y[k] e^{-2\pi i m \frac{k}{N}}$ for relevant $k$'s. For $m$ in $\{0, \ldots, N-1\}$, these are exactly the phases of the DFTs of all $N$ circular time-shifts of $x$, which are indeed all optimal.

It remains to show that the amplitudes of the DFT of $x$ (the power spectrum of $x$) can be recovered from $B$. If $x$ has nonzero mean, the power spectrum can be read off the diagonal of $B$. If $y[0] = 0$, the power spectrum can still be recovered. Indeed, consider $\log|B[k, \ell]|$:
\begin{align*}
\log|B[k, \ell]| & = \log|y[k]| + \log|y[\ell]| + \log|y[\ell - k]|.
\end{align*}
These provide linear equations in $\log|y[1]|, \ldots, \log|y[K]|$. It suffices to collect $K$ independent ones. Considering in order equations with $k = 1, \ell = 2, \ldots, K$
{
followed by $k = 2, \ell = 4$,
we get a structured linear system.
	For example, with $K = 5$:
	\begin{align*}
		\begin{bmatrix}
		2 & 1 &   &   &    \\
		1 & 1 & 1 &   &    \\
		1 &   & 1 & 1 &	   \\
		1 &   &   & 1 & 1  \\
		  & 2 &   & 1 & 
		\end{bmatrix}
		\begin{bmatrix}
		\log|y[1]| \\
		\log|y[2]| \\
		\log|y[3]| \\
		\log|y[4]| \\
		\log|y[5]|
		\end{bmatrix} & = \begin{bmatrix}
		\log|B[1, 2]| \\
		\log|B[1, 3]| \\
		\log|B[1, 4]| \\
		\log|B[1, 5]| \\
		\log|B[2, 4]|
		\end{bmatrix}.
	\end{align*}
	The determinant of the matrix is $(-1)^{K} 6$, 
proving $|y[k]|$ can be recovered from $B$. Together with $y[0] = 0$ and the phases, this is sufficient to recover $x$ up to global time shift.
}



\subsection{Proof of identity \eqref{eq:Madjk}} \label{sec:proof_identity_Madj}

Observe that $T$ can be expressed as
\begin{align} \label{eq:Tsum} 
T(z) & = \sum_{k=0}^{N-1} z[k] T_k, 
\end{align}
where $T_k$ is a circulant matrix with ones in its $k$th (circular) diagonal and zero otherwise as defined in \eqref{eq:T}.
Using~\eqref{eq:Tsum} and~\eqref{eq:Mdef},
\begin{align*}
\inner{M(z)}{X} & = \inner{W^{(2)} \circ \tilde B \circ \overline{T(z)}}{X} \\
& = \inner{T(z)}{W^{(2)} \circ \tilde B \circ \overline{X}} \\
& = \sum_{k=0}^{N-1} \inner{z[k] T_k}{W^{(2)} \circ \tilde B \circ \overline{X}} \\
& = \Re\left\{ \sum_{k=0}^{N-1} \overline{z[k]} \trace\left( T_k\transpose \left(W^{(2)} \circ \tilde B \circ \overline{X}\right) \right)\right\} \\ & = \inner{z}{M\adj(X)}.
\end{align*}
Since this must hold for all $z$ and $X$, by identification on the last line,
\begin{align*}
M\adj(X)[k] & = \trace\left( T_k\transpose \left(W^{(2)} \circ \tilde B \circ \overline{X}\right) \right).
\end{align*}

\subsection{Proof of Lemma \ref{lemma:symmetryMadj}} \label{sec:proof_lemma_symmetry_Madj}
 On one hand, note that
	\begin{align*}
	(M(z)z)[k] & = \sum_{\ell = 0}^{N-1} M(z)[k,\ell] z[\ell] \\& = \sum_{\ell = 0}^{N-1} W^2[k,\ell] \tilde B[k,\ell] \overline{z[\ell-k]} z[\ell].
	\end{align*}
	On the other hand, we have from~\eqref{eq:Madjk} and  \eqref{eq:T} that
	\begin{align*}
	(M\adj(zz^*))[k] & = \trace\left( T_k\transpose \left(W^{(2)} \circ \tilde B \circ \overline{zz^*}\right) \right) \\
	& = \sum_{\ell, \ell' = 0}^{N-1} T_k[\ell' ,\ell] \left(W^{(2)} \circ \tilde B \circ \overline{zz^*}\right)[\ell', \ell] \\
	& = \sum_{\ell = 0}^{N} W^2[\ell - k, \ell] \tilde B[\ell - k, \ell] \overline{z[\ell - k]} z[\ell] \\ &= \sum_{\ell = 0}^{N} W^2[k,\ell] \tilde B[k,\ell] \overline{z[\ell - k]} z[\ell],
	\end{align*}
	where we used the assumed symmetries of $\tilde B$ and $W$ in the last step only. The expressions match, concluding the proof.

\subsection{Proof of Lemma \ref{lem:second_order_critical_points}} \label{sec:proof_lem_second_order_critical_points}

	Since $z$ is critical, $\nabla f(z) \circ \overline{z}$ is real. It remains to show that it is positive. For this, we will use the second-order condition. Under the symmetry assumptions of Lemma~\ref{lemma:symmetryMadj} which hold a fortiori in the noiseless case, the term $\inner{\dot z}{\nabla^2 f(z)[\dot z]}$ can be developed into
	\begin{align*}
&	\inner{\dot z}{\nabla^2 f(z)[\dot z]} \\  &= \inner{\dot z}{2M(\dot z)z + 2M(z)\dot z + M(\dot z)^*z + M(z)^*\dot z} \nonumber\\
	&= 3\inner{M(z)\dot z}{\dot z} + 2\inner{M(\dot z)z}{\dot z} + \inner{M(\dot z)\dot z}{z} \nonumber\\
	&= 3\inner{z}{M^{\mathrm{adj}}\left(\dot z \dot z^*\right)} + 2\inner{M(\dot z)z}{\dot z} + \inner{M(\dot z)\dot z}{z} \nonumber\\
	&= 4\inner{M(\dot z)\dot z}{z} + 2\inner{M(\dot z)z}{\dot z},
	\end{align*}
	where the last equality follows  Lemma~\ref{lemma:symmetryMadj}.
	Hence,
	\begin{equation*}
	\inner{\dot z}{\nabla^2 f(z)[\dot z]}  =  \inner{M(\dot z)}{4z\dot z^* + 2\dot z z^*}. 
	\end{equation*}
	For some index $k$, consider the tangent vector $\dot z = (iz[k])e_k$ in $\T_z\calM$, where $e_k \in \RN$ is the $k$th canonical basis vector. Since $z$ is second-order critical, we have the inequality
	\begin{align*}
		\inner{M(\dot z)}{4z\dot z^* + 2\dot z z^*} &= \inner{\dot z}{\nabla^2 f(z)[\dot z]} \\ & \leq \inner{\dot z}{D(z) \dot z} \\&= \sum_{\ell=0}^{N-1} |\dot z[\ell]|^2 \nabla f(z)[\ell] \overline{z[\ell]}.
	\end{align*}
	Plugging in the expression of $\dot z$, this is equivalent to
	\begin{align} 
	\inner{\tilde B \circ \overline{(iz[k])}T_k}{4\overline{(iz[k])}ze_k^* + 2(iz[k])e_k z^*} & \leq \nabla f(z)[k] \overline{z[k]}.
	\label{eq:foo0}
	\end{align}
	The left hand side develops as follows:
	\begin{align*}
		&\inner{\tilde B \circ \overline{(iz[k])}T_k}{4\overline{(iz[k])}ze_k^* + 2(iz[k])e_k z^*} \\ =& \inner{\tilde B \circ T_k}{4ze_k^* + 2(iz[k])^2e_k z^*} \nonumber\\
		=& \Re\left\{\sum_{\ell, \ell' = 0}^{N-1} \overline{\tilde B[\ell',\ell]} \delta_{k, \ell - \ell'} \left( 4z[\ell']\delta_{k, \ell} - 2z[k]^2 \overline{z[\ell]} \delta_{k, \ell'} \right) \right\} \nonumber\\
 =& \Re\left\{\sum_{\ell = 0}^{N-1} \overline{\tilde B[\ell-k,\ell]} \left( 4z[\ell-k] \delta_{k, \ell} - 2z[k]^2 \overline{z[\ell]} \delta_{k, \ell-k} \right) \right\} \nonumber, 
	\end{align*}
			where in the last equality we substituted $\ell' = \ell-k$.  Using the definition of $\tilde B$ as in~\eqref{eq:bispec} we then get

\begin{align} 
		&\inner{\tilde B \circ \overline{(iz[k])}T_k}{4\overline{(iz[k])}ze_k^* + 2(iz[k])e_k z^*} \nonumber \\ 		& = \Re\left\{ 4\overline{\tilde B[0,k]} z[0] - 2\overline{\tilde B[k,2k]} z[k]^2\overline{z[2k]} \right\} \nonumber\\
		& = \Re\left\{ 4\overline{\tilde y[0]} z[0] - 2\overline{\tilde y[k]}^2 \tilde y[2k] z[k]^2\overline{z[2k]} \right\} \nonumber\\
		&   = 4\inner{\tilde y[0]}{z[0]} - 2\inner{\tilde y[k]^2 \overline{\tilde y[2k]}}{z[k]^2\overline{z[2k]}}.
		\label{eq:foo1}
\end{align}
	In particular, for $k = 0$, this simplifies to $2\inner{\tilde y[0]}{z[0]}$. Then, the inequality is
	\begin{align*}
		2\inner{\tilde y[0]}{z[0]} & \leq \nabla f(z)[0] \overline{z[0]}\\
		                         & = 2\left(M(z)z\right)[0]\overline{z[0]} + \left(M(z)^*z\right)[0]\overline{z[0]} \\
		                         & = \sum_{\ell = 0}^{N-1} 2 M(z)[0,\ell]z[\ell] \overline{z[0]} + \overline{M(z)[\ell,0]} z[\ell] \overline{z[0]} \\
		                         & = \sum_{\ell = 0}^{N-1} 2\tilde B[0,\ell]\overline{z[0]} + \overline{\tilde B[\ell, 0]} z[-\ell]z[\ell] \overline{z[0]} \\
		                         & = \tilde y[0] \overline{z[0]} \left(2N + \sum_{\ell = 0}^{N-1} \overline{\tilde y[\ell]} \overline{\tilde y[-\ell]} z[\ell] z[-\ell]\right).
	\end{align*}
	Let the sum in the right hand side be denoted by $t\in\mathbb{C}$. Clearly, $|t| \leq N$ and $2\inner{\tilde y[0]}{z[0]} \geq -2$, so
	\begin{align*}
		-2 \leq (\tilde y[0] \overline{z[0]})\left(2N + t\right).
	\end{align*}
	Since $z$ is critical, the right hand side is real, so that it is equal to either $|2N+t|$ or $-|2N+t|$  (any real number is either its absolute value or the opposite). For contradiction, let us assume it is equal to $-|2N+t|$. Then,
	\begin{align*}
		-2 \leq -|2N + t| \leq -(2N - |t|) \leq -N.
	\end{align*}
	This is impossible if $N > 2$. Hence, $(\tilde y[0] \overline{z[0]})\left(2N + t\right) = |2N + t|$, so that $\overline{\tilde y[0]} z[0] = \sign(2N+t)$. Using $|t|\leq N$, it is a simple exercise to determine that $\inner{\tilde y[0]}{z[0]} \geq \frac{\sqrt{3}}{2}$. Turning back to general $k$ and using~\eqref{eq:foo0} and~\eqref{eq:foo1}, it follows that
	\begin{align*}
		\nabla f(z)[k] \overline{z[k]} & \geq 4\inner{\tilde y[0]}{z[0]} - 2\inner{\tilde y[k]^2 \overline{\tilde y[2k]}}{z[k]^2\overline{z[2k]}} \\
		& \geq 2(\sqrt{3}-1) \\& > 0.
	\end{align*}
	Manually checking the statement for $N = 1, 2$ concludes the proof.
	

	\subsection{Proof of Lemma \ref{lem:uniqueness_real}} \label{sec:proof_uniqueness_real}
	
	The proof is identical to that of Lemma~\ref{lem:uniquness}, up to a few differences we highlight. With the same definition of the vector $u$, the identity $u[k]\overline{ u[\ell] } = \overline{ u[\ell-k] }$ still holds for $k, \ell$ such that $k, \ell, \ell - k$ are in $\{1, \ldots, K\} \cup \{N-1, \ldots, N-K\} \bmod N$. Using real symmetry, $u[-k] = \overline{u[k]}$ for all $k$, hence the identity also reads $u[k] = u[\ell] u[k-\ell]$.
	
	Using this rule for $k = 2, \ldots, K$ and fixed $\ell = 1$, we easily get $u[k] = u[1]^k$ for $k \in \{1, \ldots, K \}$. We now consider the rule for $k = N-K, \ell = K$. Since $K \geq \frac{1}{3}N$, it is clear that $k-\ell = N-2K \leq \frac{1}{3}N \leq K$. Since $K \leq \frac{N-1}{2}$, it also holds that $k - \ell \geq 1$. As a result, $u[N-K] = u[1]^{N-K}$. Likewise, using real symmetry and indexing modulo $N$, $u[N-K] = u[-K] = \overline{u[K]} = u[1]^{-K}$. Combining the two, it follows that $u[1]^N = 1$. One can then conclude as in the proof of the previous lemma.
	{The magnitudes of the DFT can also be recovered, following the same procedure as in Lemma~\ref{lem:uniquness}: if $y[0] \neq 0$, read the power spectrum off the diagonal; otherwise, obtain $|y[1]|, \ldots, |y[K]|$ via the same linear system and use $|y[N-k]| = |y[k]|$ for $k = 1\ldots K$.}


\subsection{Proof of Lemma \ref{lem:Mz}} \label{sec:proof_lem_Mz}

For ease of notation, let $A = W^{(2)} \circ \tilde B$. By the assumptions of Lemma~\ref{lemma:symmetryMadj}, we know that $A[\ell-k, \ell] = A[k, \ell]$ for all $k, \ell$. Since $A$ is now also Hermitian, we have
	\begin{align*}
		A[k, \ell] & = \overline{A[\ell, k]} = \overline{A[k-\ell, k]} = A[k, k-\ell].
	\end{align*}
	Hence, with the change of variable $\ell' = k-\ell$ and using both $z[-k] = \overline{z[k]}$ and $\dot z[-k] = \overline{\dot z[k]}$ for any $k$ since $z\in\calM_\reals$ and $\dot z\in\T_z\calM_\reals$:
	\begin{align*}
	\left( M(z)\dot z \right)[k] & = \sum_{\ell = 0}^{N-1} M(z)[k, \ell] \dot z[\ell] \\
	& = \sum_{\ell = 0}^{N-1} A[k, \ell] \overline{z[\ell-k]} \dot z[\ell] \\
	& = \sum_{\ell' = 0}^{N-1} A[k, k-\ell'] \overline{z[-\ell']} \dot z[k-\ell'] \\
	& = \sum_{\ell' = 0}^{N-1} A[k, \ell'] z[\ell'] \overline{\dot z[\ell'-k]} = M(\dot z) z.
	\end{align*}
	This concludes the proof.


\end{document}